\documentclass[sigconf]{acmart}
\usepackage{booktabs} 
\usepackage{verbatim}
\usepackage{indentfirst}
\usepackage{algorithm}
\usepackage{algorithmic}
\usepackage{subfigure}
\usepackage{graphicx}

\usepackage{makecell}

\usepackage{amsmath}
\setcopyright{rightsretained}
\setcopyright{none} 
\acmConference[Anonymous Submission to ACM CCS 2019]{ACM Conference on Computer and Communications Security}{Due 15 May 2019}{London, TBD}
\acmYear{2019}

\begin{document}
\title{Differentially Private Aggregated  Mobility Data Publication Using Moving Characteristics} 

\author{Zhili Chen, Xiaoli Kan,Shun Zhang}
\affiliation{%
  \institution{Anhui University}
  \streetaddress{No. 111, Jiulong Road}
  \city{Hefei}
  \state{Anhui}
  \country{China}
  \postcode{230601}
}
\email{zlchen@ahu.edu.cn,xl_kan@163.com,shzhang27@163.com}

\author{Lin Chen}
\affiliation{%
  \institution{University of Paris-Sud}
  \streetaddress{Office: 242, batiment 650, LRI}
  \city{Paris}
  \country{France}
}
\email{chen@lri.fr}

\author{Yan Xu, Hong Zhong}
\affiliation{%
  \institution{Anhui University}
  \streetaddress{No. 111, Jiulong Road}
  \city{Hefei}
  \state{Anhui}
  \country{China}
  \postcode{230601}
}
\email{xuyan@ahu.edu.cn,zhongh@mail.ustc.edu.cn}

\begin{abstract}
With the rapid development of GPS enabled devices (smartphones) and location-based applications, location privacy is increasingly concerned. Intuitively, it is widely believed that location privacy can be preserved by publishing aggregated mobility data, such as the number of users in an area at some time. However, a recent attack shows that these aggregated mobility data can be exploited to recover individual trajectories. In this paper, we first propose two differentially private basic schemes for aggregated mobility data publication, namely direct perturbation and threshold perturbation, which preserve location privacy of users and especially resist the trajectory recovery attack. Then, we explore the moving characteristics of mobile users, and design an improved scheme named static hybrid perturbation by combining the two basic schemes according to the moving characteristics. Since static hybrid perturbation works only for static data, which are entirely available before publishing, we further adapt the static hybrid perturbation by combining it with linear regression, and yield another improved scheme named dynamic hybrid perturbation. The dynamic hybrid perturbation works also for dynamic data, which are generated on the fly during publication. Privacy analysis shows that the proposed schemes achieve differential privacy. Extensive experiments on both simulated and real datasets demonstrate that all proposed schemes resist the trajectory recovery attack well, and the improved schemes significantly outperform the basic schemes.
\end{abstract}




\maketitle
\section{Introduction}
With the development of wireless communication and mobile positioning technologies, a growing number of mobile devices have been equipped with GPS precise positioning functions. This change makes location-based service (LBS) become increasingly popular, which is one of the most promising services for mobile users. LBS refers to the services providing information and entertainment to mobile users based on geographic locations and other information in mobile devices \cite{Lee2011Protecting}. It has been widely used in military, government industry, commerce, medical, emergency, and people's livelihood \cite{Zhou2011Location}. Some typical applications include map applications and GPS navigation (e.g., Google Maps), coupons or discount offers (e.g., MeiTuan), social communications (e.g., WeChat). However, while bringing huge benefits to individuals and society, LBS has also aroused serious privacy concerns.

When users access to LBS, they need to report their locations, which can be used to infer sensitive personal information, such as home addresses, living habits, health status and social relations, etc. Therefore, the disclosure of location information to untrusted third parties (e.g., LBS providers) will possibly lead to the abuse of personal data, and produce serious risks to user privacy. For example, according to the anonymous GPS data, one can infer an individual's home address, work unit and social relationships, and can predict the user's past, present and future locations, track an individual's whereabouts. Even indoor location information can be used to infer an individual's job role, age, hobbies (smoking or not), etc. Therefore, the protection of users' location privacy is a crucial issue \cite{zhang2015privacy}.

Previously, many schemes have been proposed to protect user location privacy in LBS. For example, Niu et al \cite{Niu2014Achieving} applied k-anonymity to the data publication in LBS, such that personal privacy cannot be identified by attackers. Wang et al \cite{Wang2018A} replaced usernames in LBS queries with temporary pseudonyms to break the link between a user's identity and its queries. In order to hide users' real locations, Niu et al \cite{Niu2014Privacy} proposed to achieve user anonymity by generating dummy locations, and using virtual locations instead of real locations. In addition, some other works proposed to solve the privacy problem in the LBS through the trajectory generalization process \cite{Stefanakis2012Trajectory}, or based on encryption methods \cite{Schlegel2015User}.

Although people have made great efforts and contributions to the privacy protection for LBS users, new privacy problems are still emerging. The recent work \cite{Xu2017Trajectory} argued that user privacy is not preserved in aggregated mobility data. They designed an attack that recovers individual trajectories from the aggregated mobility data without any prior knowledge by exploiting the uniqueness and regularity of human mobility. The accuracy of the recovered trajectories is between 73\% and 91\%. This is a fatal blow to LBS providers (mobile operators) that have published aggregated data. They were intuitively convinced that sorely the aggregated data could not reveal any privacy information of individual users, but now it becomes a serious privacy problem.

Differential privacy \cite{Dwork2006Differential} has been widely accepted as a standard for privacy protection. This technique has been developed in a series of articles \cite{Dinur2003Revealing,Blum2005Practical,Dwork2004Privacy} and finally taken sharp in \cite{Dwork2006Calibrating}. This technology can achieve any level of privacy protection, and intuitively it captures the increased risks of privacy breaches. In many cases, it can ensure a high level of privacy protection, as well as the sufficient accuracy of the database information. Differential privacy has been widely used in data publishing and data mining \cite{Li2016Differential}. Especially in the privacy protection of aggregated statistics \cite{Kellaris2014Differentially,Li2015Differentially,Fan2012Real,Zhang2013Differentially,Qian2016RescueDP}, it elegantly bounds the knowledge gain of an adversary whether a user opts in or out of a dataset.

In this paper, we propose differentially private schemes for perturbing and publishing mobile aggregated data. By analyzing the distribution of users at different times in cities \cite{Wang2015Understanding,Isaacman2012Human}, we find the moving characteristics of urban users as follows. During the daytime, users usually appear in several locations and move frequently, resulting in large $L_{1}$ distances of aggregated mobility data at adjacent time-stamps. However, in the nighttime, users are often near homes, or even in static states, causing that the $L_{1}$ distances are relatively small. In our schemes, we divide a day into daytime and nighttime two periods accordingly with differential privacy. In the daytime period when the data change rapidly, we add noise directly to the aggregated mobility data with Laplace mechanism. In the nighttime period when the data change slowly, we compare the changes with a threshold and add Laplace noise to the data accordingly. Furthermore, after perturbing the aggregated mobility data, we design a post-processing mechanism to process the noise data, which greatly improves the data utility. The contribution of this paper can be summarized as follows:

$\bullet$ To our knowledge, we are the first to design a differentially private scheme for aggregated mobility data publication, and enhance the design by taking advantage of moving characteristics of mobile users. Previous work mainly applies differential privacy to common aggregated data publication without considering the application-specific characteristics of data.

$\bullet$ For static data publication, we design the static hybrid perturbation mechanism based on users' moving characteristics, adopting different perturbation methods in different time periods of a day. Moreover, we design a post-processing mechanism to smooth the noise data, improving the data utility.

$\bullet$ For dynamic data publication, we design the dynamic hybrid perturbation by combining the static hybrid perturbation with learning and prediction, such that the mechanism first learns from the historical data, and then predicts for the current data accordingly.

$\bullet$ We fully implement our scheme, and evaluate it based on both simulated and real datasets. Experimental results show that our scheme provides strong privacy protection against the trajectory recovery attack, and ensures the utility of data as well.

The rest of this paper is organized as follows. In Section~\ref{sec:related-work}, the related work is introduced. In Section~\ref{sec:preliminaries}, the preliminaries are given. Section~\ref{sec:system-model} is the system model. We elaborate on the basic schemes in Section~\ref{sec:basic} and the improved schemes in Section~\ref{sec:improved}. We include the detailed experimental evaluation of the proposed schemes in Section~\ref{sec:exeriments} and make conclusions in Section~\ref{sec:conclusions}.

\section{Related Work}\label{sec:related-work}
In this section, we review the related work of location privacy from the following aspects.

\textbf{Traditional Protection.} There exist various traditional techniques for protecting users' locations or trajectories. Niu et al. \cite{Niu2014Achieving} adopted the k-anonymity to protect users' location privacy by preventing identification. Papers \cite{Wang2018A} and \cite{Niu2014Privacy} ensured location privacy by pseudonym replacement and dummy location generation, respectively. Stefanakis \cite{Stefanakis2012Trajectory} protected users' movement trajectories by a trajectory generalization technique. Schlegel et al. \cite{Schlegel2015User} proposed dynamic grid system to preserve users' trajectory privacy. These methods for location privacy protection normally based an ad hoc basis, lacking of a rigorous mathematical foundation. It is indicated that an adversary can still obtain sensitive information of users through some background knowledge, auxiliary information, and the like \cite{Xu2017Trajectory}. Different from this, our work leverages the notion of differential privacy, which is of rigorous theoretical foundations, and can rule out background knowledge attacks.

\textbf{DP-based Protection.} Differential privacy (DP) has been widely applied to location privacy protection. Here, we only review the most recent works for data release, which are relevant to our work. Fan et al. \cite{Fan2012Real} proposed a real-time aggregated data release scheme at event-level, FAST, by comparing the error between a prior estimate and a posterior one of perturbed data, and then conducting an adaptive sampling and perturbation. However, FAST must pre-assign the maximum times of publications and the privacy budget of each sampling point. For infinite stream data publication, Kellaris et al. \cite{Kellaris2014Differentially} proposed a $w$-event privacy framework, protecting any event sequence occurring within any sliding window of $w$ time stamps. When $w=1$, this scheme degrades into an aggregated data publication at event-level like \cite{Fan2012Real}. Li et al. \cite{Li2015Differentially} proposed publication schemes for dynamic datasets with differential privacy. They calculated distances of data at adjacent time stamps, and determined whether to release a new perturbation or an old one by comparing with a threshold. These works have well addressed the location privacy protection for various data release issues. However, to enhance the data utility, they only considered the common features of data, while did not take into account the application-specific features of data.





\textbf{Location Characterization.} There are some works capturing the characteristics of people's locations. According to \cite{Wang2015Understanding,Isaacman2012Human}, the mobile patterns of human beings during every day are quite regular. Activities are most frequent during the daytime, followed by the evening, and lest frequent at midnight. Literatures \cite{Cao2017Spatio,Chen2018The} analyzed the temporal and spatial characteristics of urban population, showing the time durations of stays in the nighttime are significantly longer than those in the daytime.
These location characteristics have been exploited to launch a location privacy attack. For instance, Xu et al. \cite{Xu2017Trajectory} designed an attack method that accurately recovers the trajectories of users based on aggregated location data using the location characteristics of human beings.
This attack has also taught us a lesson that even aggregated data reveal users' privacy. We have used the location characteristics of users in our work, and different from literature \cite{Xu2017Trajectory}, we use them to improve the data utility while protecting the location privacy.




\section{Preliminaries}\label{sec:preliminaries}
In this section, we briefly review the notion of differential privacy and the technique of linear regression.

\subsection{Differential Privacy}

Differential privacy (DP) ensures a similar function output of a dataset no matter a user's personal information is in or out of the dataset. The formal definition can be described as follows.

\begin{definition}
 (DIFFERENTIAL PRIVACY)\cite{Dwork2006Differential} A randomized mechanism $\mathcal{M}$: $\mathcal{D} \rightarrow \mathcal{O}$ satisfies $\epsilon$-DP, where $\epsilon\geq 0$, if for all sets $\textsl{O}\subseteq \mathcal{O}$, and every pair $D,D' \in\mathcal{D}$ of neighboring databases, we have
 \begin{equation}\label{definition}
   Pr[\mathcal{M}(D)\in\textsl{O}]\leq e^{\epsilon}\cdot Pr[\mathcal{M}(D')\in\textsl{O}]
 \end{equation}

\end{definition}

Differential privacy can be achieved by adding an appropriate amount of interference noise to the return value of a query function. Adding too much noise will affect the utility of the results, while too small noise will not suffice the privacy guarantee. Sensitivity is the key parameter that determines what amount of noise should be added. It refers to the maximum amount of change caused by deleting a record in the worst case.

\begin{definition}[Global Sensitivity \cite{dwork2014algorithmic}]
Let $D$ and $D'$ denote any pair of neighboring databases. The global sensitivity of a function $f$, denoted by $\Delta f$, is given as below
\begin{equation}\label{definition}
  \Delta f=\max_{D,D'}|f(D)- f(D^{\prime})|
\end{equation}

\end{definition}

Differential privacy has post-processing property, as stated in Theorem~\ref{the:post-processing}.
\begin{theorem}[Post-processing \cite{dwork2014algorithmic}]\label{the:post-processing}
Given a randomized mechanism $\mathcal{M}_1$ that satisfies $\epsilon$-DP, then for any randomized algorithm $\mathcal{M}_2$, the composition $\mathcal{M}_2(\mathcal{M}_1)$ still satisfies $\epsilon$-DP.
\end{theorem}

For sequential composition and parallel composition, differential privacy can be guaranteed by Theorem~\ref{the:sequential} and Theorem~\ref{the:parallel}, respectively.
\begin{theorem}[Sequential Composition \cite{Li2016Differential}]\label{the:sequential}
Let $\mathcal{M}_1,...,\mathcal{M}_r$ be a set of mechanisms, where $\mathcal{M}_i$ provides $\epsilon_i$-DP. Let $\mathcal{M}$ be another mechanism that executes $\mathcal{M}_{1}(D)$, $cdots$, $\mathcal{M}_{r}(D)$ using independent randomness for each $\mathcal{M}_i$, and returns the vector of the outputs of these mechanisms. Then, $\mathcal{M}$ satisfies$(\Sigma^{r}_{i=1}\epsilon_{i})$-DP.
\end{theorem}

\begin{theorem}[Parallel Composition \cite{Li2016Differential}]\label{the:parallel}
Let $\mathcal{M}_1$, $\cdots$, $\mathcal{M}_k$ be $k$ mechanisms that satisfy $\epsilon_1$-DP, $\ldots$, $\epsilon_k$-DP, respectively. For a deterministic partitioning function $f$, let $D_1,...,D_k$ be respectively the resulting partitions of excuting $f$ on dataset $D$. Publishing the results of $\mathcal{M}_1(D_1),\ldots, \mathcal{M}_k(D_k)$  satisfies $(\max_{i\in[1,...,k]}\epsilon_{i})$-DP.
\end{theorem}


In practice, there are different methods to achieve DP. The Laplace mechanism and the exponential mechanism \cite{Dwork2006Differential,Li2016Differential,Mcsherry2007Mechanism} are two basic mechanisms. The Laplace mechanism is suitable for the protection of numerical results, and the exponential mechanism is also suitable for the protection of non-numeric results.

$\mathbf{Laplace\: Mechanism.}$ The Laplace mechanism implements $\epsilon$-DP by adding random noise following the Laplace distribution to query results, as described in Theorem~\ref{the:laplace}.
\begin{theorem}\label{the:laplace}
For any function $f$: $\mathcal{D}\rightarrow \mathbb{R}^{d}$, the mechanism $\mathcal{M}$
\begin{equation}
  \mathcal{M}(\mathcal{D}) = f(\mathcal{D})+Lap(\Delta f/\epsilon)
\end{equation}
gives $\epsilon$-DP.
\end{theorem}

Here, $Lap(b)$ represents the Laplace distribution with a scale parameter $b$, and the corresponding probability density function is as follows.
\begin{equation}\label{theorem}
  p(x) = \frac{1}{2b}\exp(-\frac{|x|}{b})
\end{equation}

In many practical applications, the query results may be non-numeric (such as a scheme or a choice). In this regard, McSherry et al. \cite{Mcsherry2007Mechanism} proposed the exponential mechanism.

$\mathbf{Exponential\:Mechanism.}$ Let the output range of the mechanism $\mathcal{M}$ be $\mathcal{R}$, $q(D,r)$ be the utility function used to evaluate the pros and cons of the output $r \in \mathcal{R}$, the exponential mechanism is as follows.
\begin{theorem}
Given a utility function $q:(\mathcal{D}\times \mathcal{R})\rightarrow \mathbb{R}$ for a database $D \in \mathcal{D}$, the mechanism $\mathcal{M}$,
\begin{equation}\label{throrem}
  \mathcal{M}(D,q)=\{return\: r\: with\: probility\: \propto \exp(\frac{\epsilon q(D,r)}{2\Delta q})\}
\end{equation}
gives $\epsilon$-DP, where $\Delta q$ is the sensitivity of function $q$, and is defined as
\begin{equation}\label{throrem}
  \Delta q = \max_{r \in \mathcal{R}} \max_{D_1, D_2} |q(D_1,r) - q(D_2,r)|
\end{equation}
where $D_1$ and $D_2$ are any pair of neighboring databases.
\end{theorem}

\subsection{Linear Regression}

Linear regression \cite{Eberly2003Correlation} is one of the most basic methods of machine learning \cite{Mohri2018Foundations}. It is a linear model that makes predictions through a linear combination of attributes. It aims to find a line or a plane or a higher-dimensional hyperplane that minimizes the error between predicted values and real values. The results obtained by linear regression are well interpretable and the calculation entropy is not complicated. Linear regression can be divided into unitary regression and multiple regression. Unitary regression means that there is only one impact factor, and a familiar linear equation is involved. For multiple regression, there are multiple impact factors. In this paper, we use unitary regression to make predictions.

$\mathbf{Hypothesis \:Function.}$ The purpose of regression is to predict the target value of another numerical data from several known data. It is assumed that the features and results satisfy a linear relationship (taking the formula $h_{\theta}(x) = \theta_{0} +\theta_{1}x$ as an example in the following). This relationship is called regression equation, and it is a mathematical model chosen to fit some of the data given already.

$\mathbf{Cost \:Function.}$ Through the linear regression algorithm, we may get a lot of linear regression models, but different models have different abilities to fit the data. We need to find a linear regression model that best describes the relationship between the data. The cost function is used to describe the difference between the linear regression model and the real data. If there is no difference at all, the linear regression model fully describes the data relationship. If we need to find the best fitting linear regression model, we need to minimize the corresponding cost function. And it usually defined as $J(\theta_{0},\theta_{1}) = \frac{1}{2n}\sum_{i=0}^{n}(h_{\theta}(x_i)-y_i)^2$

$\mathbf{Gradient \:Decent.}$ In order to find the optimal hypothesis function parameters, the goal is to make the cost function take the smallest value. Gradient descent can help us find the local minimum point of a function. It can be used not only in linear regression model, but also in many other machine learning models. The idea is as follows:
(1) initialization parameters $\theta_{0}$ and $\theta_{1}$;
(2) continuously update their values at the same time to make them smaller, and finally find the minimum value point of $J(\theta_{0},\theta_{1})$.

\section{System Model}\label{sec:system-model}

In this section, we formally describe the problem definition and the attack model, and then especially review the attack proposed in \cite{Xu2017Trajectory}.

\subsection{Problem Definition}

Fig.~\ref{fig:system-model} shows the system model. In this model, we let $M$ denote the number of base stations in a city, and $N$ represent the total number of mobile users which are covered by all base stations in a city. At time stamp $t_{i}$, let $u_{m}^{i}$ denote the number of mobile users at the $m_{th}$ base station, and let $l_{n}^{i}$ denote the base station where the $n_{th}$ user is located. We define an aggregated mobility dataset as $D = \{D_i\}_{i=1}^S$, and $D_i=\{u_1^i,u_2^i,...,u_M^i\}$ representing the histogram for numbers of mobile users in all base stations at time stamp $t_i$. Without loss of generality, we assume that the time duration for data publications is a day. It is obvious to extend the time duration to any length.
For each time stamp $t_i$, we need to release a private aggregated dataset $\widetilde{D}_{i}$. And finally, we need to release a noisy version $\widetilde{D}$ for the data $D$. We summarize all notations used in this paper in Table~\ref{tab:notations}.

\begin{figure}
\includegraphics[scale=0.38]{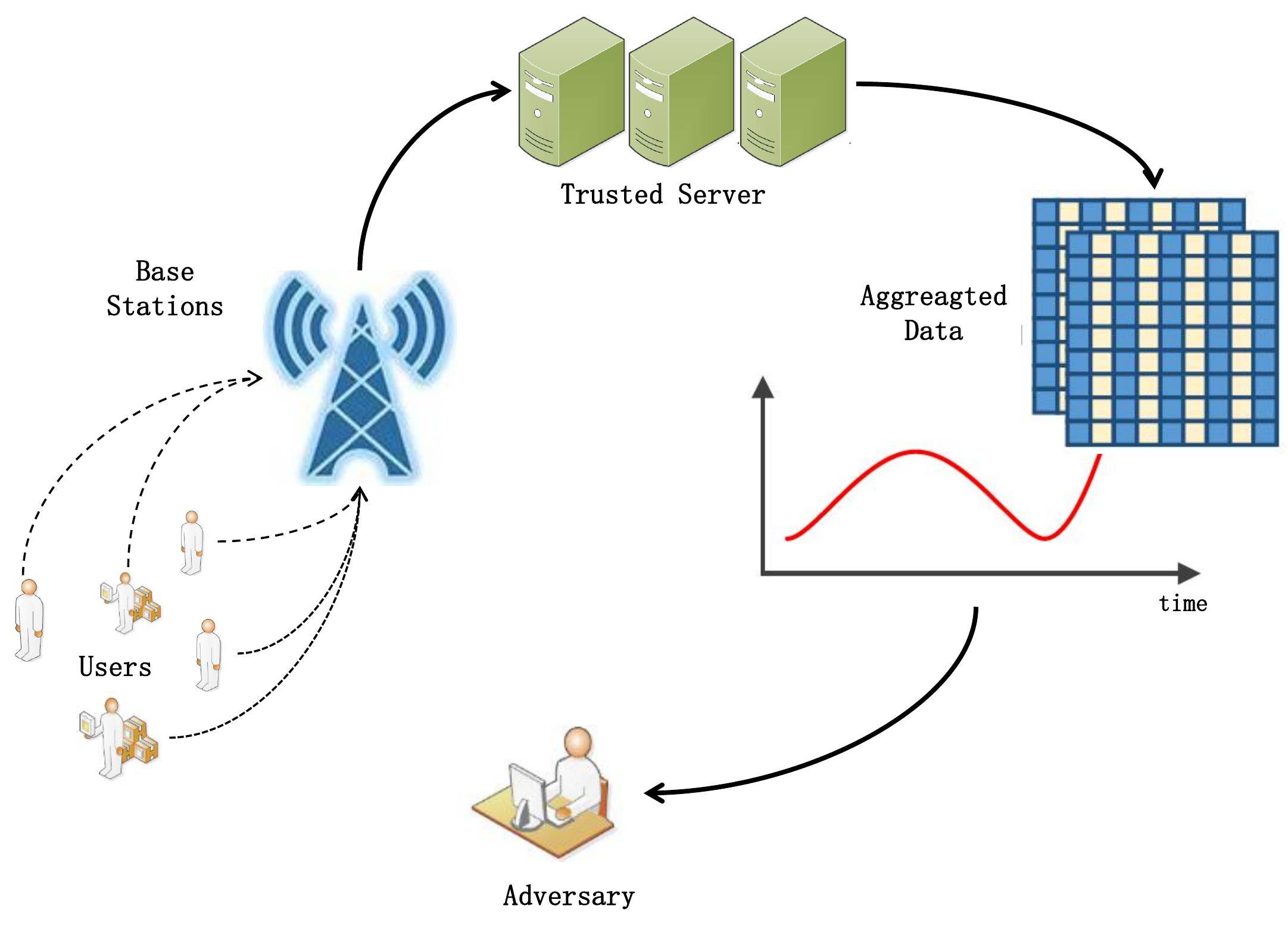}
\caption{System Model}\label{fig:system-model}
\end{figure}

\begin{table}[tbp]\small
\centering
\begin{tabular}{|c|c|}
\hline
Notation & Descriptioin\\
\hline
$D$ & Original aggregated dataset \\
\hline
$D'$ & Neighboring dataset \\
\hline
$\widetilde{D}$ & DP aggregated dataset \\
\hline
$D_{i}$ & Original aggregated dataset at time point $t_{i}$\\
\hline
$\widetilde{D}_{i}$ & DP aggregated dataset at time point $t_{i}$\\
\hline
$M$ & Number of base stations\\
\hline
$N$ & Number of mobile users\\
\hline
$u^{i}_{m}$ & \makecell{Number of mobile users at the $m_{th}$ base station \\ at time point $t_{i}$}\\
\hline
$l^{i}_{n}$ & \makecell{Base station location of the $n_{th}$ user \\ at time point $t_{i}$}\\
\hline
$\epsilon$ & Total privacy budget\\
\hline
$S$ &Number of time stamps\\
\hline
$\epsilon_{1}$ & Privacy budget for comparing with threshold\\
\hline
$\epsilon_{2}$ & Privacy budget for perturbaing with threshold \\
\hline
$\epsilon_{l}$ & Left privacy budget during threshold perturbation \\
\hline
$T$ & threshold\\
\hline
$\epsilon_{s}$ & Privacy budget for time division\\
\hline
$\epsilon_{d}$ & Privacy budget for direct perturbation \\
\hline
$\epsilon_{t}$ & Privacy budget for threshold perturbation \\
\hline
$\hat{t}_1$ & the first time division point \\
\hline
$\hat{t}_2$ & the second time division point \\
\hline
\end{tabular}
\caption{notations}\label{tab:notations}
\end{table}

\subsection{Attack Model}

We assume that the service provider is trustable, and it carries out the data publication process honestly. Also, the service provider is responsible for doing its best to protect users' location privacy. The data then release to some data consumers, which may be adversaries. We assume that the adversaries have not any prior information of the targeted data. They do their best to analyze the data released, and try to infer as much as possible location information of single mobile users.


Under this attack model, literature \cite{Xu2017Trajectory} has proposed an attack method based on the regularity and uniqueness of mobile users' movements. This attack aims to recover mobile users' individual trajectories from the the aggregated mobility data with an unsupervised method. It iteratively associates the same user's movement records in adjacent time stamps and gradually recovers the entire trajectory. In each time stamp, the attack process can be divided into two steps. Firstly, the possibility that next position $l_{n}^{i}$ belongs to a given trajectory can be estimated by utilizing the characteristics of human mobility. Secondly, an optimal solution is obtained by the idea of linear sum assignment, so that it can link the mobile users' trajectories with the next record of mobility.

\section{Basic Schemes}\label{sec:basic}
In this section, we propose two basic schemes, the direct perturbation scheme based on the Laplace mechanism \cite{Dwork2006Differential} and the threshold perturbation scheme based on the sparse vector technique \cite{Dwork2006Differential}. We also propose a consistency post-process mechanism, which can be applied to the basic schemes, as well as the improved schemes in next section, for preserving the consistency of noisy aggregated mobility data.

\subsection{Direct Perturbation}

To publish the aggregated mobility data $D = \{D_i\}_{i=1}^S$ with $\epsilon$-differential privacy, where $D_i=\{u_1^i,u_2^i,...,u_m^i\}$, a naive idea is to perturb the histogram $D_i$ at each time stamp $t_i$ with the privacy budget $\epsilon/S$ by applying Laplace mechanism. Then, in terms of the sequential composition theorem, we get a differentially private data publication scheme as required. Alg.~\ref{alg:direct-perburbation} shows the procedure of the direct perturbation mechanism.

\renewcommand{\algorithmicrequire}{\textbf{Input:}}  
\renewcommand{\algorithmicensure}{\textbf{Output:}}  

\begin{algorithm}[htb]
\caption{Direct Perturbation Mechanism} \label{alg:direct-perburbation}
\begin{algorithmic}[1]

\REQUIRE Dataset $D = \{D_i\}_{i=1}^S$, privacy budget $\epsilon$, histogram sensitivity $\Delta_H$
\ENSURE Noisy dataset $\widetilde{D} = \{\widetilde{D}_i\}_{i=1}^S$

\STATE Set $b = S \cdot \Delta_H/\epsilon$;
\FOR{$i=0$ to $S$}
\STATE Generate $n^i = (n^i_1,n^i_2...,n^i_m)$ with i.i.d. $Lap(b)$;
\STATE Set $\widetilde{D}_i = D_i + n^i $;
\ENDFOR
\RETURN $\widetilde{D}=\{\widetilde{D}_i\}_{i=1}^S$;

\end{algorithmic}
\end{algorithm}


Note that in Alg.~\ref{alg:direct-perburbation}, we regard two aggregated mobility datasets $D$ and $D'$ that differ in only a user's trajectory as neighboring sets. More specifically, if $D'$ contains the same user trajectories as $D$ except excluding one, we say $D$ and $D'$ are neighboring sets. According to this definition of neighboring sets, we can compute the histogram sensitivity $\Delta_H$ of the direction perturbation is 1.

The direct perturbation mechanism simply achieves differential privacy. However, as the increase of the number of time stamps $S$, the Laplace noise grows larger, and the data utility becomes worse. The reason is that this perturbation mechanism perturbs the histogram of each time stamp independently, without considering the dependences between histograms of adjacent time stamps.

\subsection{Threshold Perturbation}

An improvement on the naive idea is that, we do not have to perturb the histogram of very time stamp, but just perturb the histograms when there are significant changes. It would be a great waste of privacy budget if the data change little, but they are perturbed independently for each time stamp. For this reason, we can use the threshold perturbation based on the sparse vector technique. The threshold perturbation has been adapted from \cite{Li2015Differentially}, and the sparse vector technique takes the form of that in \cite{dwork2014algorithmic}.

Alg.~\ref{alg:threshold} shows the threshold perturbation mechanism.
In the algorithm, we first divide the privacy budget $\epsilon$ into two parts, $\epsilon_1$ and $\epsilon_2$, where the former is for comparing with the threshold and the latter is for perturbing data (Line~\ref{line:epsilon12}). Next, the first histogram $D_1$ is directly perturbed with $Lap(c\Delta/\epsilon_1)$ noise (Line~\ref{line:perturb1}), the first noisy threshold $\widetilde{T}_1$ is set (Line~\ref{line:tildeT1}), and the count of perturbations $cnt$ is initialized (Line~\ref{line:cnt1}). Then, the algorithm loops to perturb the data from time stamp $2$ to $S$, repeatedly (Lines~\ref{line:for} to \ref{line:endfor}). There are three cases for this loop to deal with. First, when the privacy budget runs out, all left histograms are released with the last perturbed one (Lines~\ref{line:ifcnt} to \ref{line:endifcnt}). Second, if there are some privacy budget left when perturbing the last histogram, all the left budget $\epsilon_{l}$ is used to released the histogram (Lines~\ref{line:ifi} to \ref{line:endifi}). Third, in the middle of the loop, each histogram is released by comparing a noisy distance and a noisy threshold, and the noisy histogram can be a newly perturbed one or the previously perturbed one accordingly (Lines~\ref{line:dist} to \ref{line:endifd}). Here, the threshold $T$ is a fixed value determined by the distribution of users at different time stamps.

\renewcommand{\algorithmicrequire}{\textbf{Input:}}  
\renewcommand{\algorithmicensure}{\textbf{Output:}}  

\begin{algorithm}[htb]
\caption{\small Threshold Perturbation Mechanism} \label{alg:threshold}
\begin{algorithmic}[1]\small

\REQUIRE Dataset $D = \{D_i\}_{i=1}^S$, privacy budget $\epsilon$, threshold $T$, distance sensitivity $\Delta_D$, cutoff point $c$

\ENSURE Noisy dataset $\widetilde{D} = \{\widetilde{D}_i\}_{i=1}^S$

\STATE Set $\epsilon_1 = \alpha \epsilon$, $\epsilon_2 =(1-\alpha)\epsilon$;\label{line:epsilon12}

\STATE Perturb $D_1$ with $Lap(c\Delta_D/\epsilon_2)$ and get $\widetilde{D}_1$;\label{line:perturb1}

\STATE Set $\widetilde{T}_1 = T + Lap(c\Delta_D/\epsilon_1)$;\label{line:tildeT1}

\STATE Set $cnt = 1$;\label{line:cnt1}

\FOR{$i=2$ to $S$}\label{line:for}

\IF{$cnt \ge c$}\label{line:ifcnt}
\STATE Set $\widetilde{D}_i = \widetilde{D}_{i-1}$;
\STATE \textbf{continue};
\ENDIF\label{line:endifcnt}

\IF{$i==S$}\label{line:ifi}
\STATE Perturb $D_S$ with $Lap(\Delta_D/\epsilon_{l})$ and get $\widetilde{D}_S$;
\STATE \textbf{continue};
\ENDIF\label{line:endifi}

\STATE Set $\widetilde{d}_i = \emph{Dist}(\widetilde{D}_{i-1},D_i) + Lap(\frac{2c\Delta_D}{\epsilon_{1}})$;\label{line:dist}

\IF {$\widetilde{d}_i \ge \widetilde{T}_{cnt}$}
\STATE Perturb $D_i$ with $Lap(c\Delta_D/\epsilon_2)$ and get $\widetilde{D}_i$;
\STATE Set $cnt = cnt + 1$;
\STATE Set $\widetilde{T}_{cnt}=T + Lap(c\Delta_D/\epsilon_1)$;
\ELSE
\STATE Set $\widetilde{D}_i = \widetilde{D}_{i-1}$;
\ENDIF\label{line:endifd}

\ENDFOR\label{line:endfor}

\RETURN $\widetilde{D}$;

\end{algorithmic}
\end{algorithm}

The threshold perturbation mechanism considers the data dependences between histograms at adjacent time stamps. Specifically, it uses the previous noisy histogram as the release of the current original histogram if there are only small changes. Therefore, if there is adequate  data dependence, this mechanism can greatly reduce the use of privacy budget, and thus improve the performance of data publication. Note that we can compute $\Delta_D = 2$ due to the definition of our neighboring sets.

\subsection{Consistency Post-processing}

After perturbation, the released values may be fractional or negative. In order to preserve the validity of the data, we design a consistency post-processing mechanism. The mechanism deals with the noisy data $\widetilde{D}$, meeting the two following conditions: (a) the value of each element is a non-negative integer, (b) the total numbers of users after post-processing remains the same as that of the perturbed data. This processing also has the impact that it can reduce the error caused by the introduction of noise and improve the data utility. We detail how to implement the consistency post-process mechanism in Alg.~\ref{alg:consistency}.

\renewcommand{\algorithmicrequire}{\textbf{Input:}}  
\renewcommand{\algorithmicensure}{\textbf{Output:}}  

\begin{algorithm}[htb]
\caption{Consistency Post-process Mechanism} \label{alg:consistency}
\begin{algorithmic}[1]

\REQUIRE Noisy dataset before post-processing $\widetilde{D}$
\ENSURE Noisy dataset after post-processing $\widetilde{D}$

\STATE Set $\widetilde{D} = round(\widetilde{D})$;\label{line:round}
\FOR{$i=1$ to $S$}\label{line:fori}

\STATE Set $\emph{cnt} = 0$;
\FOR {$j=1$ to $M$}
\IF{$\widetilde{u}_j^{i}<0$}
\STATE Set $\emph{cnt} = \emph{cnt} + |\widetilde{u}_j^{i}|$;
\STATE Set $\widetilde{u}_j^{i} = 0$;
\ENDIF
\ENDFOR

\WHILE {$\emph{cnt} > 0$ \AND $|\widetilde{D}_i|>0$}
\STATE 	Find $j \in_R \{1,2,\cdots,M\}$;\
\IF{$\widetilde{u}_j^{i} > 0$}
\STATE Set $\widetilde{u}_j^{i} = \widetilde{u}_j^{i} - 1$;
\STATE Set $\emph{cnt} = \emph{cnt} - 1$;
\ENDIF
\ENDWHILE

\ENDFOR\label{line:endfori}

\RETURN $\widetilde{D}$

\end{algorithmic}
\end{algorithm}

In Alg.~\ref{alg:consistency}, the algorithm first rounds the value of each element of noisy data $\widetilde{D}$ to an integer (Line~\ref{line:round}). Then a FOR loop is executed to do the post-processing for the histograms of all time stamps (Lines~\ref{line:fori} to \ref{line:endfori}). The loop first sets all negative values of the histogram of the current time stamp to 0, and then make appropriate adjustments to other positive values, such that the sum of each noisy histogram remains the same.

We can show that the post-processing mechanism improves the data utility. Equivalently, we show that the post-processing reduces the $L_1$ distance between a noisy histogram and its original one. First, it is obvious that the rounding has no impact upon the distance on average. Second, the non-negative processing reduces the distance. The reason is as follows. We keep the sum of the noisy histogram the same before and after post-processing, and the change from negatives to zeros reduces the distance, while the decrease for positives has the same possibility (i.e., $1/2$) to reduce and increase the distance. Thus, on average, the distance is reduced and the data utility is improved.

\section{Improved Schemes}\label{sec:improved}
In this section, we first analyze the moving characteristics of mobile users based on the aggregated mobility data. Then, we improve the basic schemes by combining them according to the moving characteristics, and design the static and dynamic hybrid perturbation mechanisms for static and dynamic data publications, respectively. Finally, we analyze the privacy of both schemes.

\subsection{Characteristic Analysis}

\begin{figure}
  \centering
  \subfigure[Synthetic Dataset]{
    \includegraphics[scale=0.24]{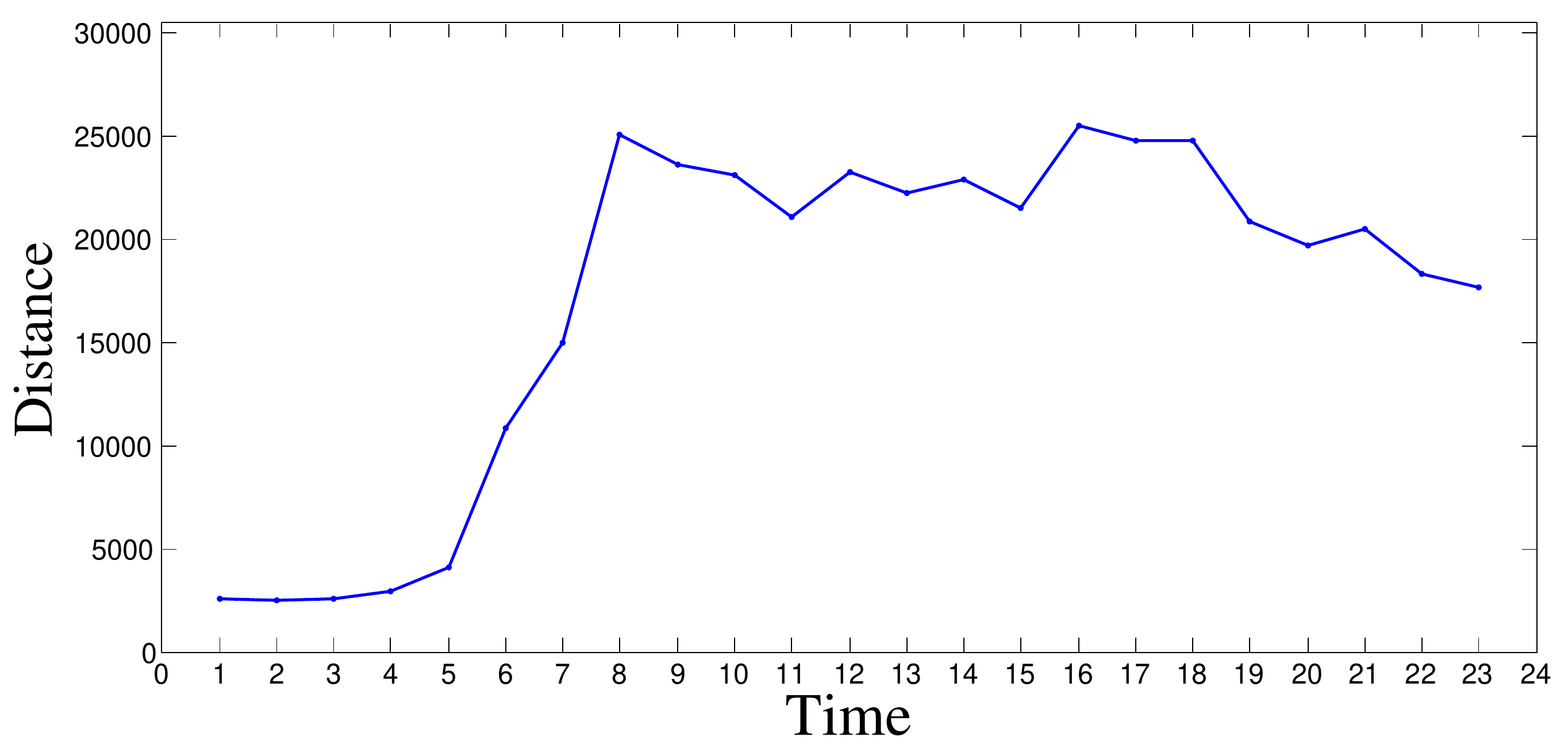}
  }
  \subfigure[Taxi Dataset]{
    \includegraphics[scale=0.24]{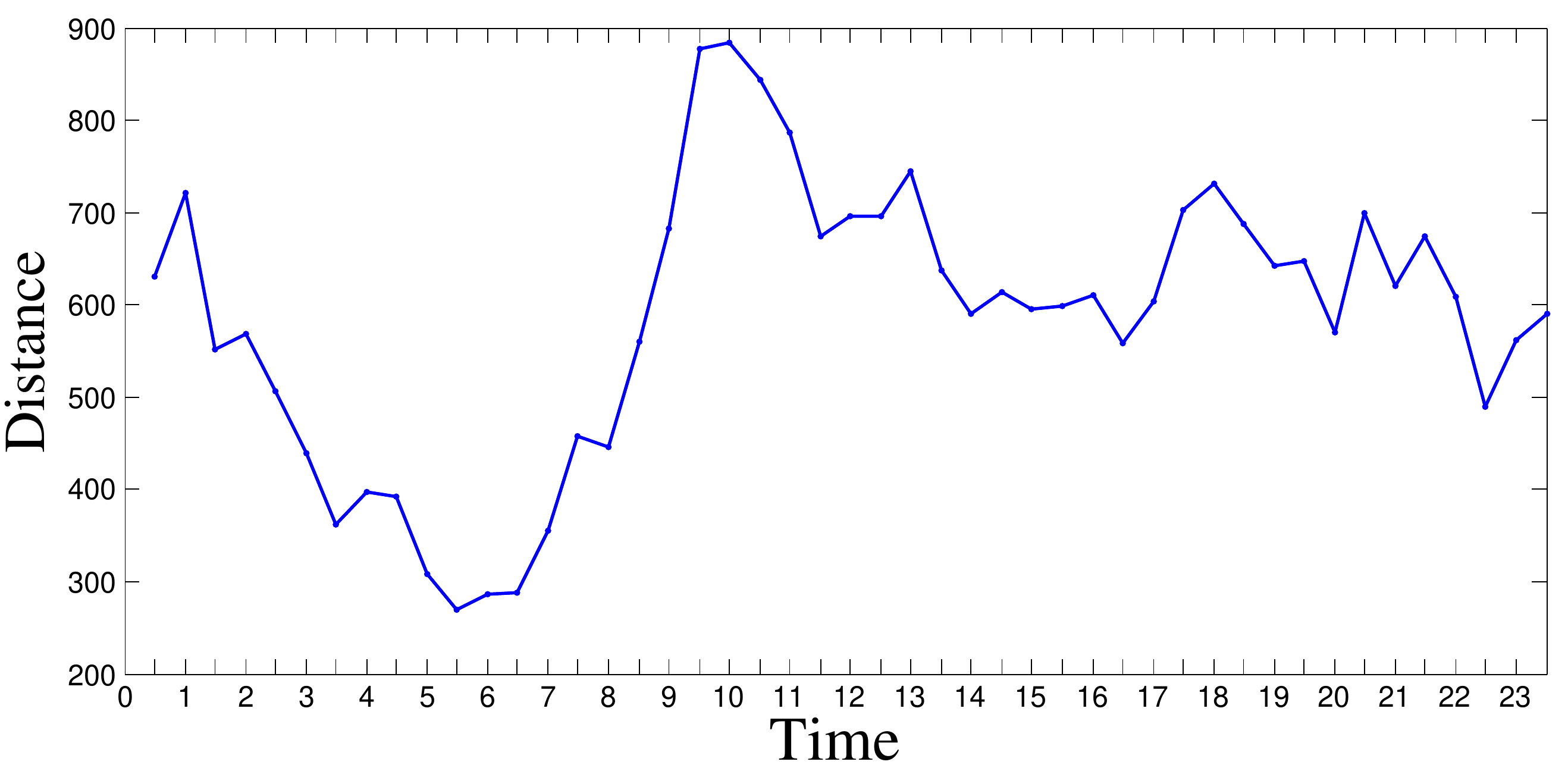}
  }
  \caption{The Distances at Adjacent Time stamps on Two Dataset}\label{fig:moving}
\end{figure}

Literatures \cite{Cao2017Spatio,Chen2018The} give a detailed analysis of temporal and spatial characteristics of urban population group activities: at different periods of a day, there are remarkable differences in the duration of time that users stay in certain locations. Specifically, the durations of stays in the nighttime are significantly longer than those in the daytime. That is, people are moving more frequently in the daytime.  In Fig.~\ref{fig:moving}, we analyze two datasets used in our experiments by computing and plotting the $L_1$ distances between histograms at adjacent time stamps in a day. The distance values represent the moving frequencies of mobile users. We found that the two datasets roughly accord with the reported moving characteristics of human beings. Specifically, during the daytime, the distances of data are large, and so is their variation. Conversely, during the night time, the distances and their variation are small.


We already know that using direct perturbation to inject noise directly into the data at each time stamp may lead to a poor utility, since the data dependence is not considered. Though the threshold perturbation can improve the data utility if there is adequate data dependence, it may also suffer a loss of data utility if there is little data dependence. The reason is that the threshold perturbation mechanism needs to compare distances and a threshold with differential privacy before publishing a histogram with a previous or a new noisy one, and this consumes extra privacy budget. Then, if users are moving frequently, and thus the resulted data are of little dependence, the threshold perturbation mechanism can cause even worse data utility than the direct perturbation mechanism. On the other hand, from the analysis of moving characteristics above, we found that the data of the daytime are of little dependence, while those of nighttime are of significant one. Therefore, our main idea is that we can divide the data into daytime and nighttime two parts, and adopt direct perturbation mechanism for publishing the daytime part and the threshold perturbation for publishing the nighttime part, respectively. Fig.~\ref{fig:framework} shows the framework of the improved schemes.
\begin{figure}
\includegraphics[scale=0.33]{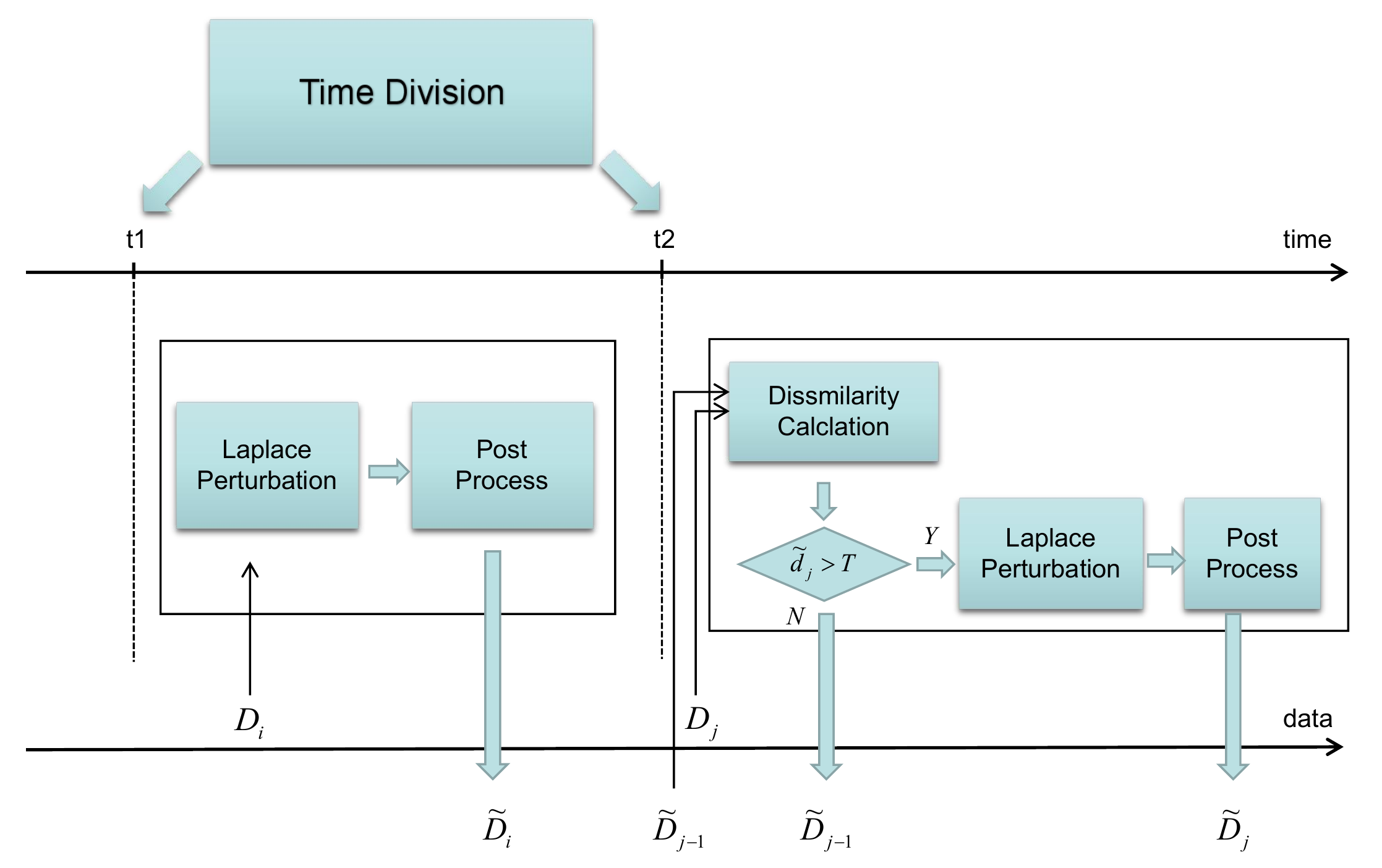}
\caption{The Framework of the Improved Schemes}\label{fig:framework}
\end{figure}

\subsection{Static Hybrid Perturbation}

Based on our main idea, we first design a hybrid perturbation for static aggregated mobility data, which are entirely obtained before publication. The main challenge to design the static hybrid perturbation mechanism is how to divide the data of a day into daytime and nighttime parts. It is noticeable that a fixed division does not preserve differential privacy. This can be illustrated by a critical user, whose trajectory will change the division points if included. Therefore, we need to perform the data division in a differentially private way, and then perturb each division of data according to their moving characteristics.

Specifically, the static hybrid perturbation selects two time division points $\hat{t}_1$ and $\hat{t}_2$ through the exponential mechanism, such that the histograms between these two time stamps change frequently, forming the daytime part of data, and then the mechanism perturbs the daytime part with direct perturbation, while perturbs other histograms, which forms the nighttime part of data, with threshold perturbation. A time division for the synthetic dataset is illustrated by Fig.~\ref{fig:division}. Then, the key is to how design a good utility function that determines how appropriate a pair of time division points are.

\begin{figure}
    \includegraphics[scale=0.24]{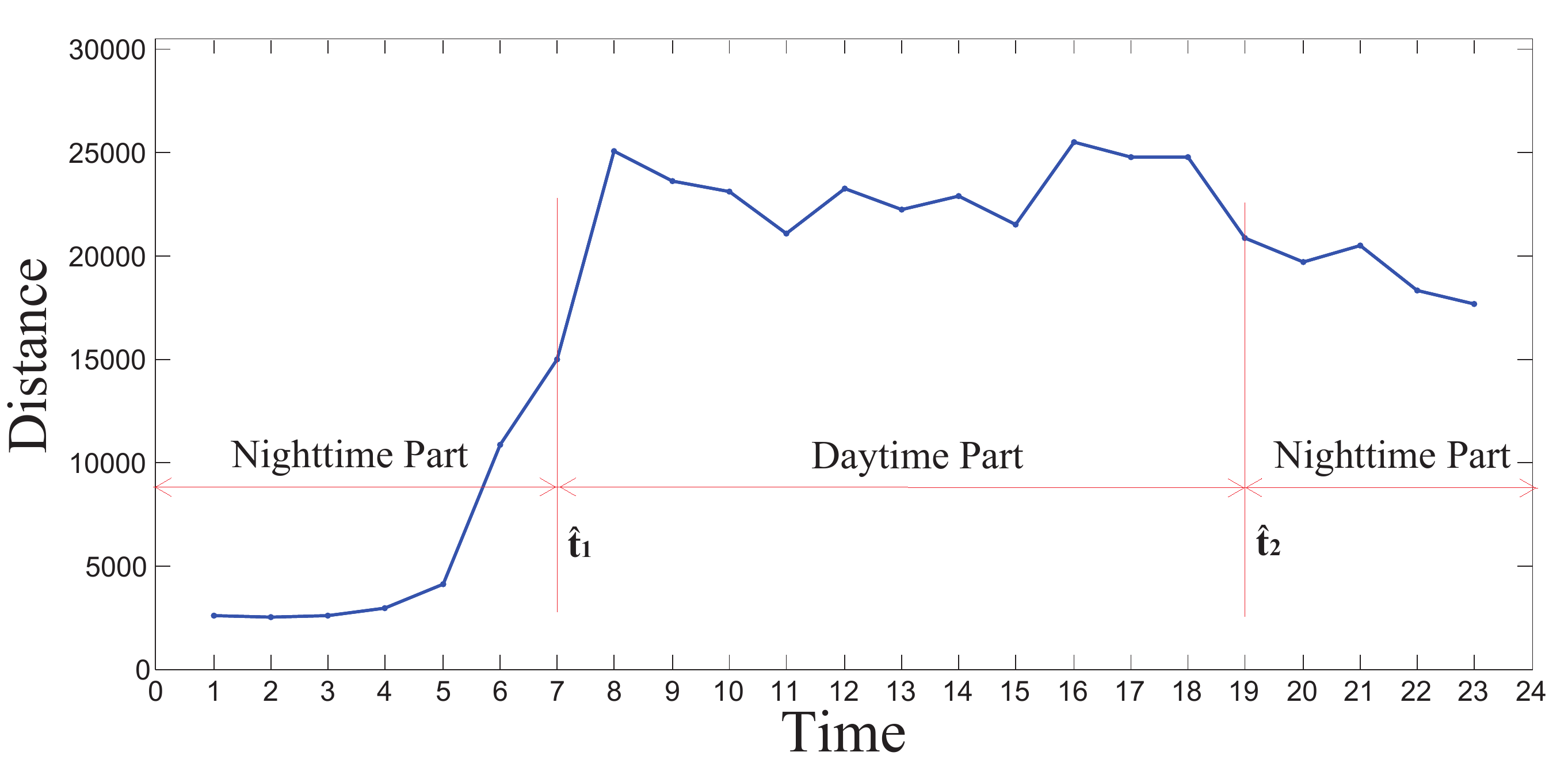}
  \caption{A Time Division for Synthetic Dataset}\label{fig:division}
\end{figure}

Observing Fig.~\ref{fig:moving}, the selection of time division points should consider following factors. First, the average distance of histograms between the two division points should be large, which represents that the histograms change greatly overall during the selected period. Second, the variance of the histograms should be small, indicating that few small distances are included. Third, the length of the selected time interval should be large enough, otherwise it is less meaningful for utility improvement. For example, in Fig.~\ref{fig:moving}, we hope that $\hat{t}_1$ and $\hat{t}_2$ fall into the working time (7:00$\sim$10:00) in the morning and the off duty time (16:00$\sim$19:00) in the afternoon, respectively. Therefore, we define the utility function as:
\begin{equation}\label{equ:utility}
U(i,j) = \log_\alpha(j-i+1)\times\frac{aveDis}{simVar}
\end{equation}
where $i$ is the first time division point, $j$ is the second one, and $i \leq j$. The related parts are described as follows.

$\mathbf{aveDis}$: It represents the average of $L_1$ distances between histograms at adjacent time stamps during the time $i$ and $j$. The larger the average distance is, the faster the histograms change. Therefore, the histograms in this time period are more suitable for direct perturbation. In our context, the average distance $aveDis$ between time stamps $i$ and $j$ is calculated by
\begin{equation}\label{equ:avedis}
aveDis = \frac{1}{j-i+1}\sum_{t=i}^{j-1}L_1(D_{t},D_{t+1})
\end{equation}

$\mathbf{simVar}$: It is a variance representing the degree of dispersion for distances during the time period between $i$ and $j$. When the variance is small and the average is large, there are equally large distances in the period, and this is ideal for the application of direct disturbance. The calculation of $simVar$ is as follows
\begin{equation}\label{equ:simvar}
simVar = \frac{1}{j-i+1}\sum_{t=i}^{j-1}|aveDis -L_1(D_{t},D_{t+1})|
\end{equation}
where when $i=j$ or $simVar < 0$, we simply set $simVar=1$.

$\mathbf{Parameter~\alpha}$: Sorely considering the two factors mentioned above does not necessarily result in a reasonable selection of division points.  For instance, when division points $i$ and $j$ are very close, and the corresponding distances are large, it then leads to that $simVar$ is small and $aveDis$ is large. Nevertheless, this is not a good selection for division points, since the time duration is too short. Therefore, we need another factor $\log_\alpha(j-i+1)$ with $\alpha > 1$ to ensure the sufficient length of the selected time interval. Empirically, $\alpha$ should be set close to the length of the interval.

$\mathbf{Sensitivity}$: The sensitivity of the utility function is $2\log_\alpha S$, where $S$ is the total number of time stamps in a day. The sensitivity can be calculated as follows. First, the sensitivity of $L_1$ distances between histograms at adjacent time stamps is at most 2, and hence the sensitivity of $aveDis$ is also 2. Second, $simVar \ge 1$ due to its definition. Third, its obvious that $j-i+1 \le S$. To sum up, we get the sensitivity $2\log_\alpha S$.

Alg.~\ref{alg:shybrid} depicts our static hybrid perturbation mechanism. The algorithm first divides $\epsilon$ into three parts, $\epsilon_s$, $\epsilon_{d}$ and $\epsilon_{t}$, for time division, direct perturbation and threshold perturbation, respectively (Line~\ref{line:budget}). Then, the utility values of $U(i,j)$ are computed in terms of Eq.~\eqref{equ:utility} (Line~\ref{line:computeU}), and the division points $\hat{t}_1$ and $\hat{t}_2$ are selected by the exponential mechanism (Lines~\ref{line:computeP} to \ref{line:select}). That is, we divide the histograms of a day into daytime part (i.e., the histograms at time stamps from $\hat{t}_1$ to $\hat{t}_2$), and the nighttime part (i.e., other histograms). Finally, the algorithm publishes the histograms at time stamps from $\hat{t}_1$ to $\hat{t}_2$ with direct perturbation mechanism, while publishes other histograms with threshold perturbation mechanism.

\renewcommand{\algorithmicrequire}{\textbf{Input:}}  
\renewcommand{\algorithmicensure}{\textbf{Output:}}  

\begin{algorithm}[htb]
\caption{Static Hybrid Perturbation Mechanism} \label{alg:shybrid}
\begin{algorithmic}[1]

\REQUIRE Dataset $D = \{D_i\}_{i=1}^S$, privacy budget $\epsilon$, threshold $T$, cutoff point $c$, sensitivity $\Delta$
\ENSURE Noisy dataset $\widetilde{D} = \{\widetilde{D}_i\}_{i=1}^S$

\STATE Divide privacy budget $\epsilon$ into $\epsilon_s$, $\epsilon_{d}$ and $\epsilon_{t}$;\label{line:budget}
\STATE Compute $U(i,j)$ for all $1 \le i \le j \le S$ by Eq.~\eqref{equ:utility};\label{line:computeU}
\STATE Compute $p_{ij} = \exp(\frac{\epsilon_{s} U(i,j)}{2\Delta})$, for all $1 \le i \le j \le S$;\label{line:computeP}
\STATE Normalize $\{p_{i,j}\}_{i \le j}$ such that $\sum_{i \le j} P_{i,j} = 1$;
\STATE Set $\mathbf{Pr} = \{p_{i,j}\}_{i \le j}$;
\STATE Select $(\hat{t}_1, \hat{t}_2)$ in terms of distribution $\mathbf{Pr}$;\label{line:select}
\STATE Call DirectPerturbation($\{D_i\}_{[\hat{t}_1,\hat{t}_2]}$, $\epsilon_{d}$);\label{line:direct}
\STATE Call ThresholdPerturbation($\{D_i\}_{[1,S]\backslash [\hat{t}_1,\hat{t}_2]}$, $\epsilon_{t}$, $c$);\label{line:threshold}

\RETURN $\widetilde{D}$

\end{algorithmic}
\end{algorithm}

It is worth noting that static hybrid perturbation is suitable for the aggregated mobility data which change greatly in some period of time, and change little other times. This quite accords with the moving characteristics of human beings in a day.

\subsection{Dynamic Hybrid Perturbation}

The static hybrid perturbation mechanism only works for static aggregated mobility data, whose histograms are all available before publication. Sometimes, we need to publish dynamic data whose histograms arrive on the fly, and the histograms after the current time stamp are not available. We propose the dynamic hybrid perturbation mechanism to deal with this situation.

The main challenge for dynamic hybrid perturbation is how we can divide the time of a day at the beginning of the publication, when we do not even know any histogram of the day. The idea is that we can make use of historical data, and leverage machine learning techniques to predict the time division. Specifically, we perform the time division in a differentially private manner for previous several days based on the historical data, and obtain the noisy division points, then adopt linear regression to train a model and predict the time division for the current day. The dynamic hybrid perturbation mechanism is described in Alg.~\ref{alg:dhybrid}.

\renewcommand{\algorithmicrequire}{\textbf{Input:}}  
\renewcommand{\algorithmicensure}{\textbf{Output:}}  

\begin{algorithm}[htb]
\caption{Dynamic Hybrid Perturbation Mechanism} \label{alg:dhybrid}
\begin{algorithmic}[1]

\REQUIRE Historical Data $\bar{D}_1$, $\bar{D}_2$,...,$\bar{D}_H$, current dataset $D$, budget $\epsilon$, threshold $T$, cutoff points $c_1$, $c_2$, sensitivity $\Delta$
\ENSURE Noisy dataset $\widetilde{D} = \{\widetilde{D}_i\}_{i=1}^S$

\STATE Divide privacy budget $\epsilon$ into $\epsilon_{d}$, $\epsilon_{t1}$ and $\epsilon_{t2}$;\label{line:budget2}
\FOR{$h=1$ to $H$}\label{line:forh}
\STATE Divide time for $\bar{D}_h$ and get division points $t_1^h$ and $t_2^h$ with budget $\epsilon$, following Lines~\ref{line:computeU} to \ref{line:select} of Alg.~\ref{alg:shybrid};
\ENDFOR\label{line:endforh}
\STATE Perform linear regression based on $\{t_1^h, t_2^h\}_{h=1}^H$, and predict division points as $\hat{t}_1$ and $\hat{t}_2$ for the current day;\label{line:regression}

\STATE Call ThresholdPerturbation($\{D_i\}_{[1,\hat{t}_1-1]}$, $\epsilon_{t1}$, $c_1$);\label{line:threshold1}
\STATE Call DirectPerturbation($\{D_i\}_{[\hat{t}_1,\hat{t}_2]}$, $\epsilon_{d}$);\label{line:directD}
\STATE Call ThresholdPerturbation($\{D_i\}_{[\hat{t}_2+1, S]}$, $\epsilon_{t2}$, $c_2$);\label{line:threshold2}

\RETURN $\widetilde{D}$
\end{algorithmic}
\end{algorithm}

In Alg.~\ref{alg:dhybrid}, the privacy budget is first divided into three parts, one for direction perturbation and other two for threshold perturbation (Line~\ref{line:budget2}), and the division points for historical data are computed with differential privacy, following the time division part of Alg.~\ref{alg:shybrid} (Lines~\ref{line:forh} to \ref{line:endforh}). Note that we assume that the historical data of each past day are independent of one another, and they are also independent of the current data when computing the historical time division points. We argue that this is reasonable, since the information released by the time division is very summary and it is actually hard to infer the dependence of this information on any individual. Next, the linear regression is trained based on the division points of historical data, and the resulted model is used to predict the time division for the current day (Line~\ref{line:regression}). Then, the data of the current day are perturbed with direct perturbation and threshold perturbation according the time division (Lines~\ref{line:threshold1} to \ref{line:threshold2}).

Now, we sketch the linear regression process as follows. We first set two hypothesis functions $T_1$ and $T_2$ corresponding to $t_1$ and $t_2$, respectively. The expressions of the hypothesis functions $T_1$ and $T_2$ are as follows:
\begin{align*}
T_1(x_1) = \theta_{1}^{0} +\theta_{1}^{1} x_1
\end{align*}
\begin{align*}
T_2(x_2) = \theta_{2}^{0} +\theta_{2}^{1} x_2
\end{align*}

Next, we use the least square method to define our cost functions $J(\theta)$. The cost functions are:
\begin{align*}
J(\theta_{1}^{0},\theta_{1}^{1}) = \frac{1}{2H}\sum_{h=1}^{H}(T_1(x_{1,h})-t_{1,h})^2
\end{align*}
\begin{align*}
J(\theta_{2}^{0},\theta_{2}^{1}) = \frac{1}{2H}\sum_{h=1}^{H}(T_2(x_{2,h})-t_{2,h})^2
\end{align*}

Let $\theta = (\theta^0_k, \theta^1_k)$, $x^{(i)} = (1, x_{k,i})$, and $t^{(i)}=t_{k,i}$, for $k=1$ or $2$. Then, we use gradient descent to iteratively optimize our $\theta$. The iterative formula of $\theta$ is as follows:
\begin{align*}
\theta:=\theta+\beta (t^{(i)}-T(x^{(i)}))x^{(i)}
\end{align*}
where $\beta$ is the learning rate.

Finally, we get the trained model and predict the time division points accordingly.

Note that, the result of linear regression determines the time division. Therefore, it does not impact on the differential privacy achieved, but it do impact on the data utility resulted.

\subsection{Privacy Analysis}
We analyze the privacy of the proposed schemes in this paper as follows.

\begin{theorem}\label{theorem1}
The direct perturbation mechanism preserves $\epsilon$-differential privacy.
\end{theorem}
\begin{proof}
The direct perturbation mechanism mainly applies the Laplace mechanism $S$ times sequentially, each of which with the privacy budget $\epsilon/S$, and then applies the consistency post-process. Therefore, it is straightforward that the mechanism achieves $\epsilon$-differential privacy due to the sequential composition and the post-processing property.
\end{proof}

\begin{theorem}\label{theorem1}
The threshold perturbation mechanism preserves $\epsilon$-differential privacy.
\end{theorem}

\begin{proof}
The threshold perturbation mechanism applies the sparse vector technique and allocates the privacy budget $\epsilon$ in two aspects. The first is comparing with the threshold, which consumes the budget $\epsilon_1/c$ at most $c$ times. The second is publishing the histograms with distances greater than the threshold, with the budget $\epsilon_2/c$ also at most $c$ times. If there are exactly $c$ times, then the total privacy budget used is $\epsilon_1/c \times c + \epsilon_2/c \times c = \epsilon$. Otherwise, according to Lines~\ref{line:ifi} to \ref{line:endifi} of Alg.~\ref{alg:threshold}, the total privacy budget used is also $\epsilon$. Therefore, the mechanism preserves $\epsilon$-differential privacy.
\end{proof}

\begin{theorem}\label{theorem1}
The static hybrid perturbation mechanism preserves $\epsilon$-differential privacy.
\end{theorem}
\begin{proof}
The static hybrid perturbation mechanism includes three parts, the time division, the direction perturbation and the threshold perturbation. The first part is in fact an exponential mechanism, and achieves $\epsilon_s$-differential privacy. The differential privacy of the other two parts has been analyzed above, thus the two parts achieve $\epsilon_d$- and $\epsilon_t$- differential privacy, respectively. Since the three parts are sequentially composed, and $\epsilon_s+ \epsilon_d + \epsilon_t = \epsilon$, the mechanism achieves $\epsilon$-differential privacy.
\end{proof}

\begin{theorem}\label{theorem1}
The dynamic hybrid perturbation mechanism preserves $\epsilon$-differential privacy.
\end{theorem}
\begin{proof}
The dynamic hybrid perturbation mechanism includes two parts, time division and data perturbation. The first part is the time divisions for the past $H$ day based on the historical data. We have assume that the data of each day are independent of each other, thus this part is the parallel composition of $H$ time divisions of data, each with $\epsilon$-differential privacy, and thus totally also achieves $\epsilon$-differential privacy. The differential privacy of the second part composing three perturbations has been analyzed above, and they achieve $\epsilon_{t1}$-, $\epsilon_d$- and $\epsilon_{t2}$- differential privacy, respectively. Since the three perturbations are sequentially composed, and $\epsilon_{t1}+ \epsilon_d + \epsilon_{t2} = \epsilon$, the second part achieves $\epsilon$-differential privacy. Finally, the two parts are parallel composed, and thus the mechanism achieves $\epsilon$-differential privacy.
\end{proof}

\section{Experiments}\label{sec:exeriments}

We fully implement the four proposed schemes, direct perturbation (dir-P), threshold perturbation (th-P), static hybrid perturbation (sta-HP) and dynamic hybrid perturbation (dy-HP), and carry out experiments over two datasets to compare their performances. In default, we set $\alpha=12$. We also use the attack proposed in \cite{Xu2017Trajectory} to perform trajectory recovery on perturbed aggregated data, and compare the accuracies of trajectories recovered. We implement the four schemes in MATLAB. All the experiments are performed on a machine with Intel(R) Core(TM) i5-6500 CPU 3.20 GHZ and 8GB RAM, running Windows 7.

\subsection{Datasets}

We conduct our experiment with two datasets. The first dataset is a real dataset about Taxi-Drive trajectories \cite{Jing2011Driving,Jing2010T}. The dataset contains the GPS trajectories of 10,357 taxis in Beijing from February 2 to February 8, 2008. The total number of points in the dataset is about 15 million, and the total distance of the trajectory reaches 9 million kilometers. The average sampling interval is approximately 177 seconds and the distance is approximately 623 meters. Each file in the data set (named by the taxi ID) contains the trajectory of a taxi. Each record in this dataset contains 4 attributes, ID, time, longitude, and  latitude. We divided this region by a grid (500m$\times$500m), which constitutes a small grid area of about 8000, and counts the number of taxis per grid every 30 minutes to form aggregate data.

For the second dataset, we simulate and generate a trajectory dataset for urban crowd movements. According to the literature \cite{Cao2017Spatio,Chen2018The}, the spatial and temporal characteristics of urban population group activities are as follows. The number of a user's stays in a day is very limited, 97.7$\%$ of users do not have more than 4 stays, and the average number of stays per person per day is about 2.1 (maybe home and workplace). At different time periods, there is a difference in the time durations users stays, and the stay time durations during the night time are significantly longer than those during the day. And the rule of crowd movement revealed in the literature \cite{Xu2017Trajectory}: Each user's mobile mode has continuity, regularity and uniqueness. We use the Beijing urban area as the geographical location to simulate the generation of about 8000 communication base station locations (500 meters $\times$ 500 meters grid). We generate the initial distribution of 130,000 people based on the characteristics of population distribution. Then, According to the distribution characteristics and user movement patterns in different time periods, the track points of each user in continuous time are generated. Similarly, each record contains four attributes, ID, time, longitude, and latitude. We select the interval time of 1 hour to aggregate the user trajectory data for 19 hours. The usability of the simulated data can be verified by the attack method in \cite{Xu2017Trajectory} later.

\subsection{Attack Evaluation}

We randomly select 1000 user records from the taxi-drive dataset and the synthetic dataset, and generate two smaller aggregated mobility datasets. Then we use the attack method proposed in \cite{Xu2017Trajectory} to perform trajectory recovery experiments on these two datasets. Fig.~\ref{fig:recovery}(a) shows the trajectory recovery on the synthetic dataset. When there are 10 points on a trajectory, the accuracy of the trajectory recovery can reach 78.86$\%$. As the number of trajectory points increases, the accuracy decreases, but the average accuracy can still reach 69.51$\%$. This result is basically consistent with the experimental conclusions in \cite{Xu2017Trajectory}, and this also proves that the user trajectory dataset generated by our simulation is feasible in this scheme. Fig.~\ref{fig:recovery}(b) shows the trajectory recovery on the real taxi dataset. When the trajectory of each taxi has 5 trajectory points, the recovery accuracy can reach 59.4$\%$. The reason why the trajectory recovery accuracy is lower than that of the synthetic dataset is that the moving characteristics of taxis lack specific rules like mobile users. Taxis often run everywhere randomly, while the moving of mobile users are normally of regularity and uniqueness characteristics, which are also the basis of the attack method.
\begin{figure}[t]
  \centering
  \subfigure[Synthetic Dataset]{
    \includegraphics[width=1.5in]{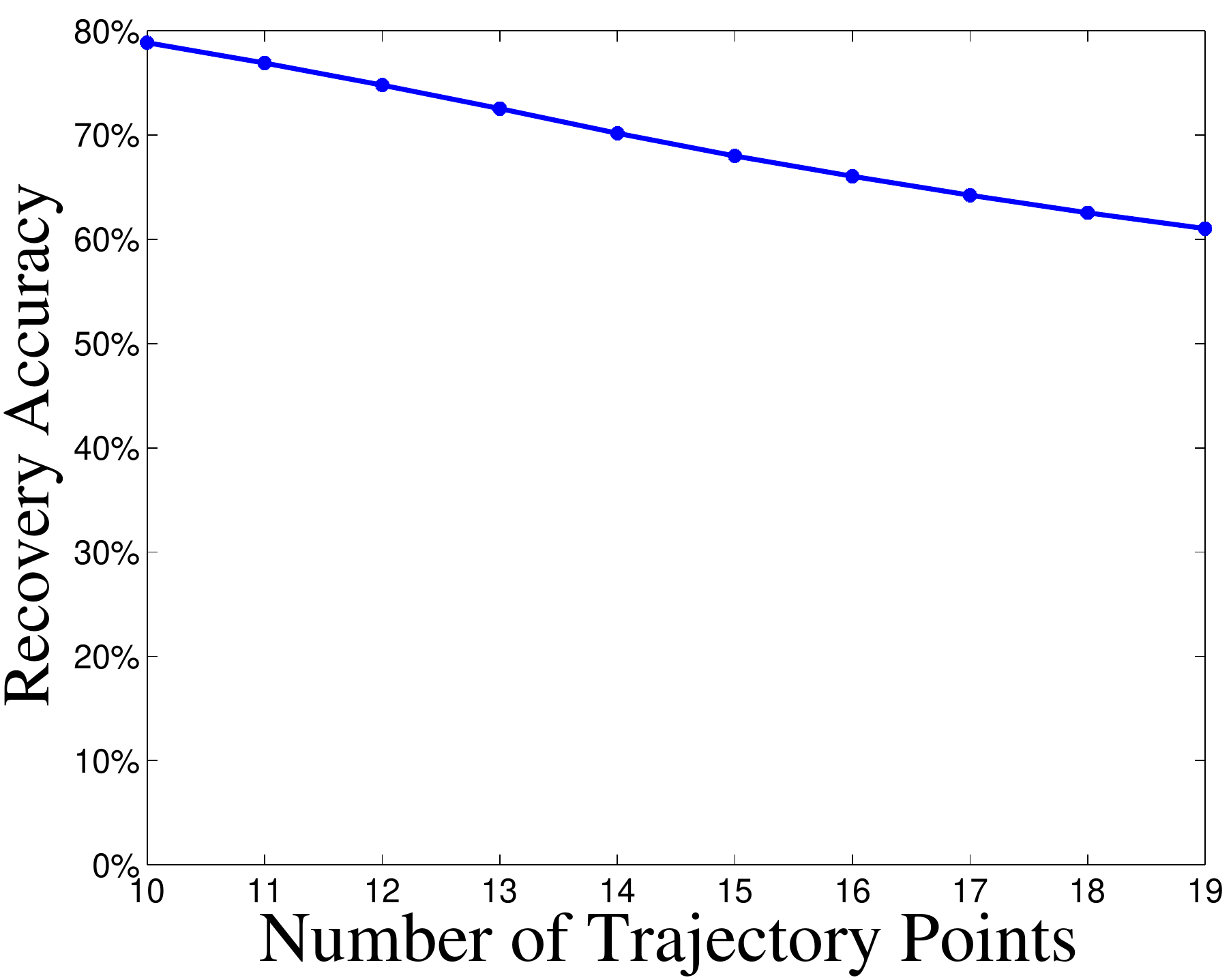}
  }
  \subfigure[Taxi Dataset]{
    \includegraphics[width=1.5in]{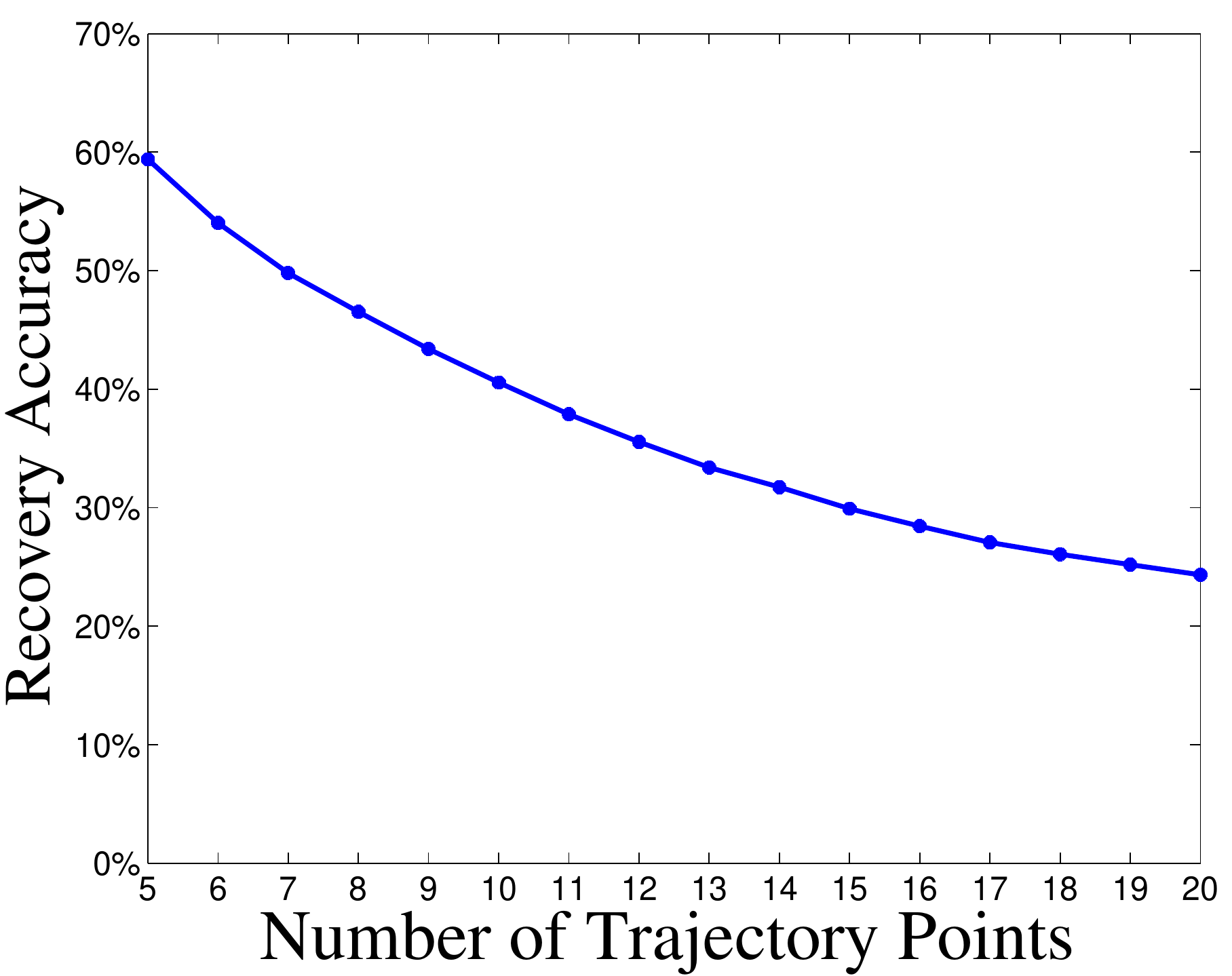}
  }
  \caption{Trajectory recovery on two dataset}\label{fig:recovery}
\end{figure}

We select the synthetic data with 19 trajectory points (whose recovery accuracy is 61.02$\%$ without DP) and the taxi data with 10 trajectory points (whose recovery accuracy is 40.56$\%$ without DP), and compare the performances of trajectory recovery with various privacy budgets under the proposed schemes. As shown in Fig.~\ref{fig:recoveryDP}, the proposed schemes have the same effect on protecting the aggregated data from being accurately restored to user trajectories. Especially on the synthetic dataset, the accuracy of trajectory recovery has dropped from 61.2$\%$ to between 10$\%$ and 20$\%$. Besides, the recovery accuracy has improved with the increase of the privacy budget, but the effect is not obvious.
\begin{figure}[t]
  \centering
  \subfigure[Synthetic Dataset]{
    \includegraphics[width=1.5in]{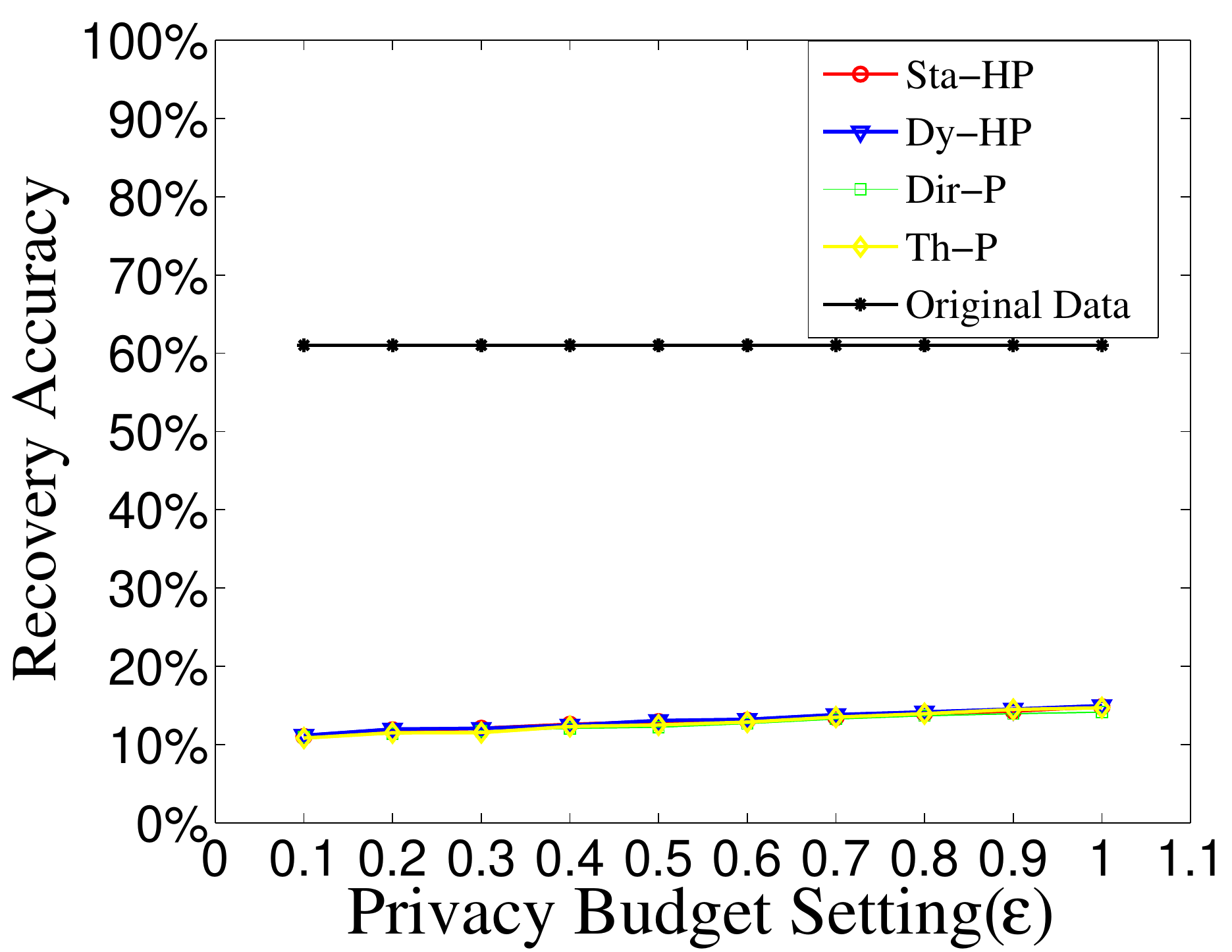}
  }
  \subfigure[Taxi Dataset]{
    \includegraphics[width=1.5in]{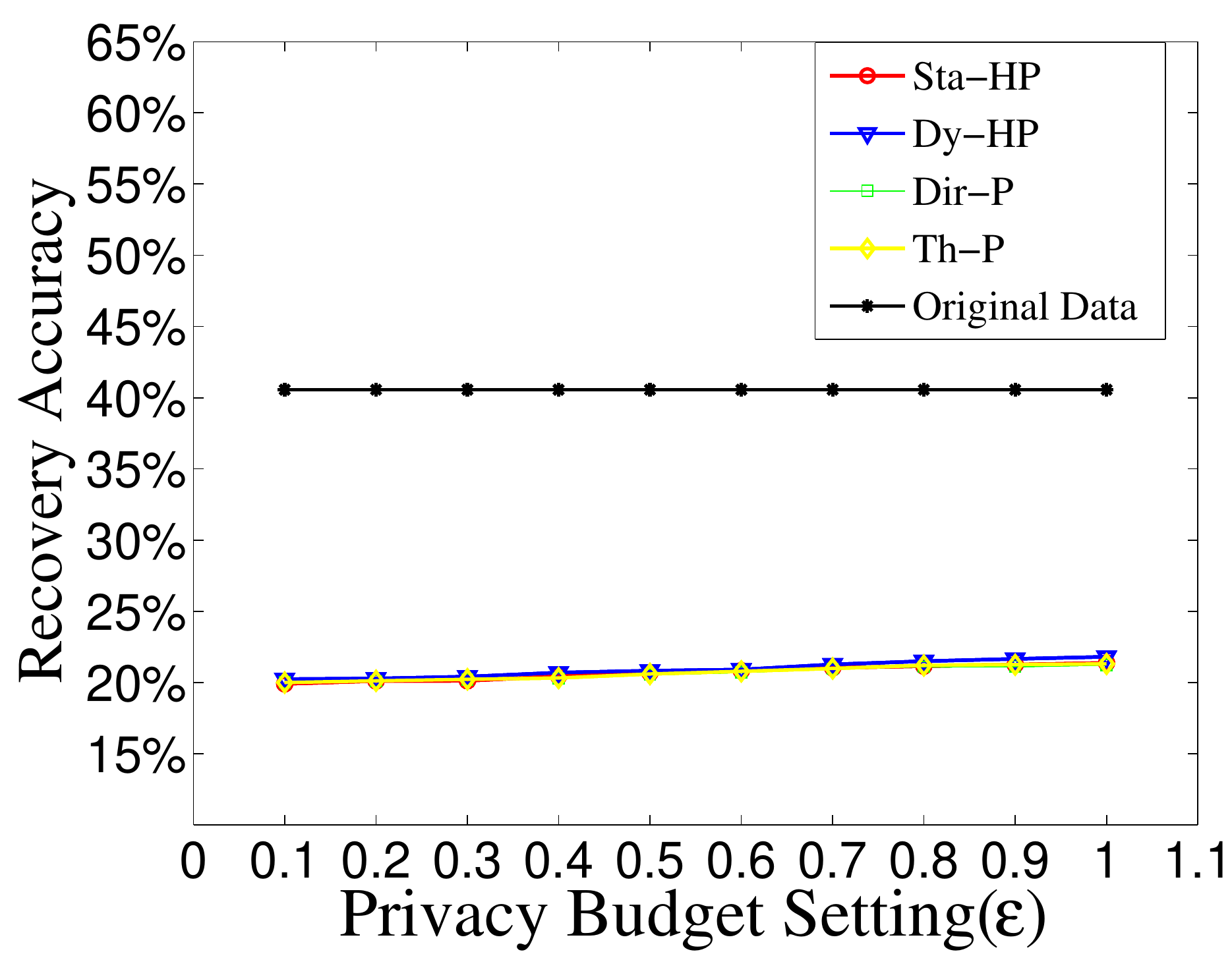}
  }
  \caption{Trajectory recovery after DP on two dataset}\label{fig:recoveryDP}
\end{figure}

\subsection{Performance Evaluation}

We use Mean Absolute Error(MAE) and Mean Relative Error (MRE) as the utility metric to evaluate the performance of the proposed schemes. Let $X = \{x_{1},\cdots,x_{t},\cdots\}$ denote the raw time series and $Y = \{y_{1},\cdots, y_{t},\cdots\}$ denote the noisy time series. In addition, let $x_{t} = \{r_{t,1},...,r_{t,M}\}$ and $y_{t} = \{r'_{t,1},...,r'_{t,M}\}$ denote raw aggregated data and noisy aggregated data at time stamp $t$, respectively. The MAE and MRE for this aggregated data at time stamp $t$ are:
\begin{equation}\label{theorem}
  \textbf{MAE}(x_{t},y_{t}) = \frac{1}{M}\sum_{i=1}^{M}|r_{t,i} - r'_{t,i}|
\end{equation}
\begin{equation}\label{theorem}
  \textbf{MRE}(x_{t},y_{t}) = \frac{1}{M}\sum_{i=1}^{M}\frac{|r_{t,i} - r'_{t,i}|}{\max(\gamma,r'_{i})}
\end{equation}
For the bound $\gamma$, we set its value to 0.001 to avoid the possibility that the denominator is 0. In our experiments, we first calculate the MAE and MRE for each aggregated data at each time stamp and then work out the average of all aggregated data as the final results.

\textbf{Utility vs Threshold T.} As shown in Fig.~\ref{fig:synthetic-threshold} and Fig.~\ref{fig:taxi-threshold}, we set $\epsilon$ to 0.8 and  compare the utility of four schemes when the threshold value varies. We let $\bar{T}$ denote the average of $L_1$ distances for adjacent time stamps of the current day. We can see that the data utility for direction perturbation keeps exactly the same, and that for the threshold perturbation decreases obviously with the decrease of the threshold. This is because the direction perturbation is dependent on the threshold, and for threshold perturbation, when the threshold becomes smaller, more histograms need to be perturbed with Laplace mechanism, then the privacy budget for each histogram perturbed is smaller, and finally the noisy added grows higher. Note that threshold perturbation consume a part of privacy budget for comparing with the threshold, and the privacy budget for publication is smaller then that of direct perturbation. We can also see that the data utility for the improved schemes keep roughly the same. The reason is that when the threshold becomes small, for the improved schemes, only the publication of nighttime part can be affected, since it uses threshold perturbation. While the nighttime part is suitable for threshold perturbation, and thus the impact on data utility is limited.

\begin{figure}
  \centering
  \subfigure[MAE]{
    \includegraphics[width=1.5in]{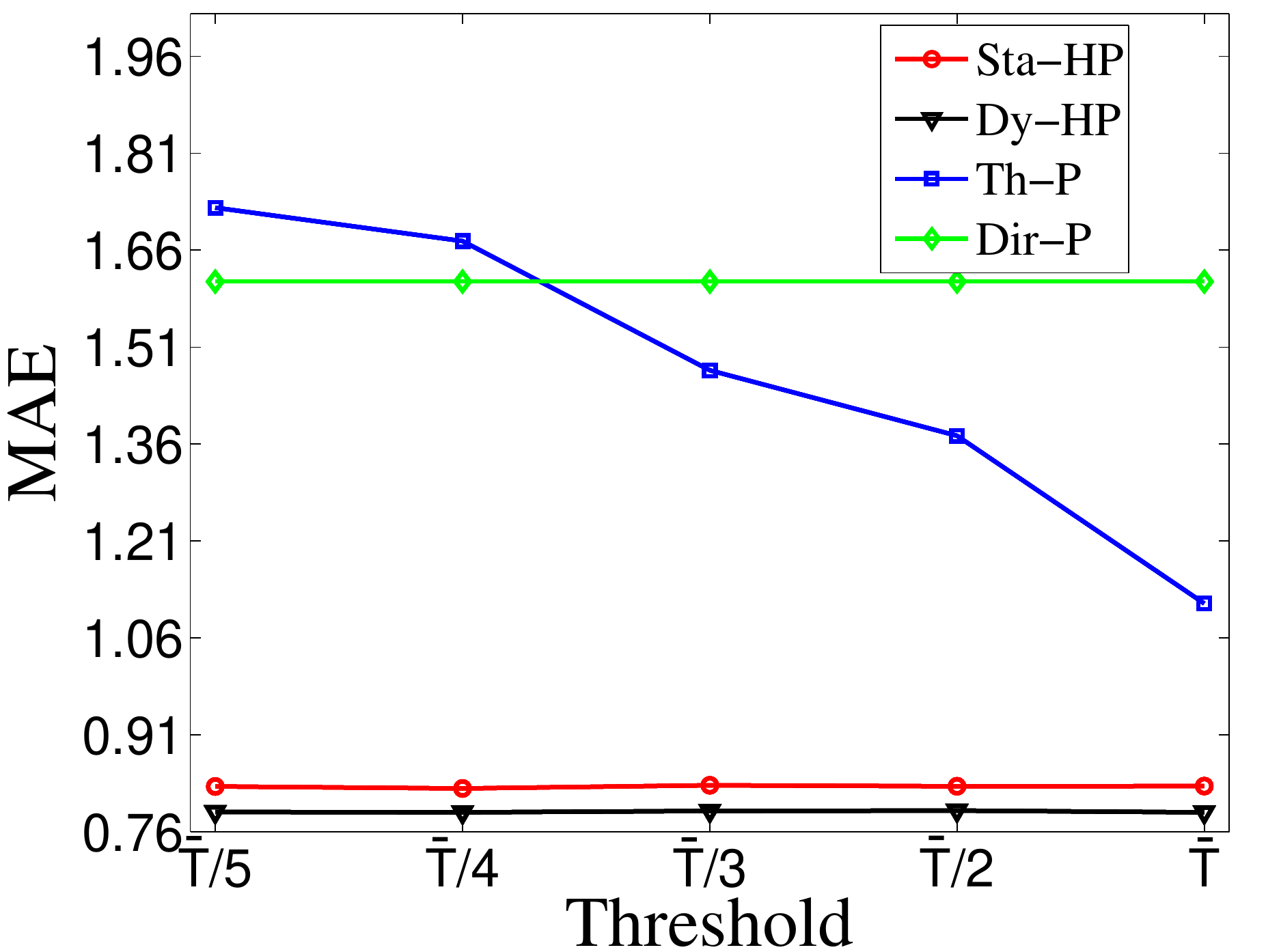}
  }
  \subfigure[MRE]{
    \includegraphics[width=1.5in]{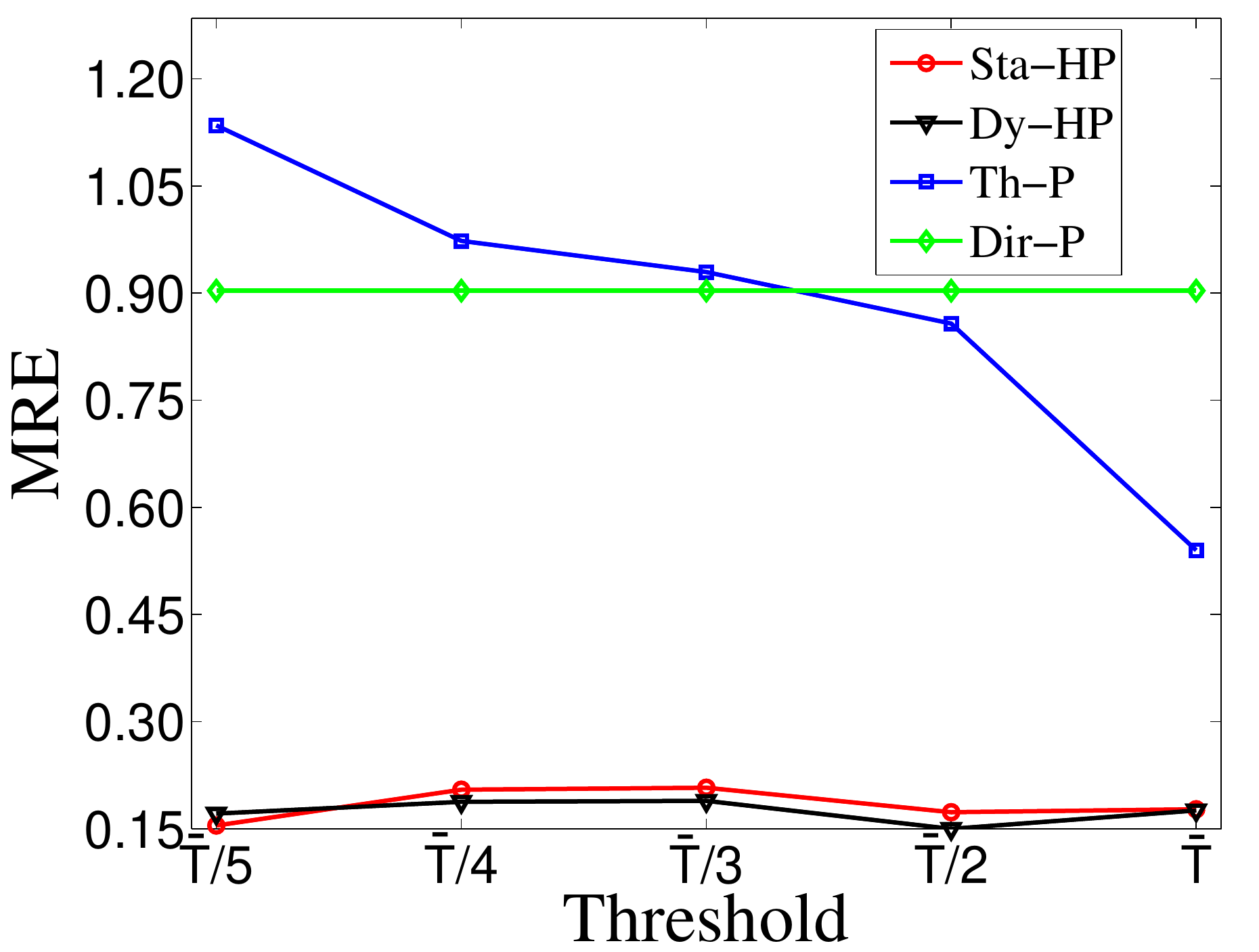}
  }
  \caption{Utility vs  $T$  on synthetic dataset}\label{fig:synthetic-threshold}
\end{figure}
\begin{figure}
  \centering
  \subfigure[MAE]{
    \includegraphics[width=1.5in]{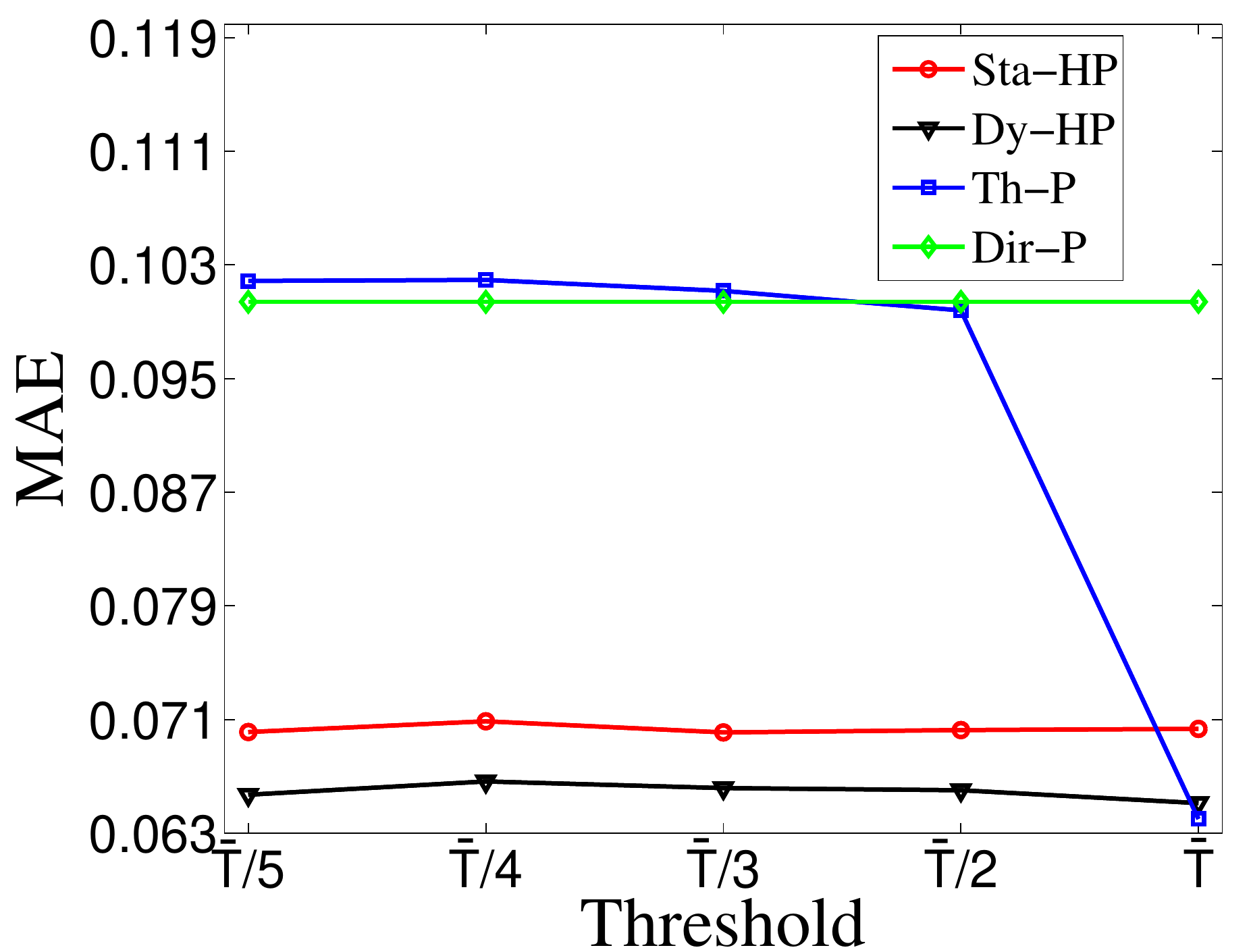}
  }
  \subfigure[MRE]{
    \includegraphics[width=1.5in]{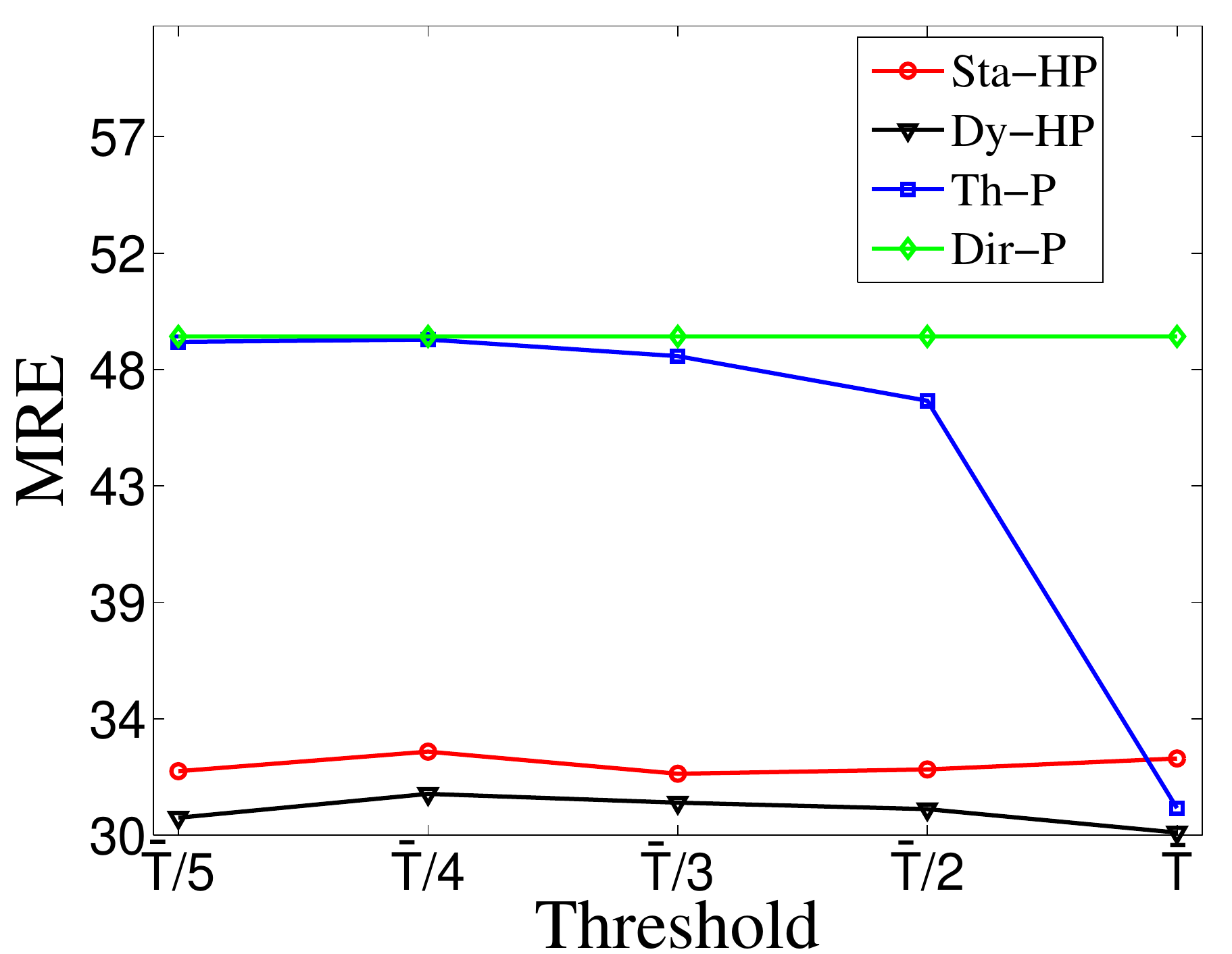}
  }
  \caption{Utility vs  $T$  on taxi dataset}\label{fig:taxi-threshold}
\end{figure}

\textbf{Utility vs Privacy.} As shown in Fig.~\ref{fig:synthetic} and Fig.~\ref{fig:taxi}, we compare MAE and MRE values of the proposed schemes when the threshold is set $\bar{T}/4$, and the privacy budget $\epsilon$ changes from 0.1 to 1 on two datasets. Here, $\bar{T}$ is defined the average of distances of the current day. We can observe that with the gradual increase of the privacy budget, the MAE and MRE values of all schemes are gradually reduced. This accords with the concept of differential privacy well. As the privacy budget increases, the amount of noise injected into data decreases, and thus the utility of data rises. On both synthetic and texi datasets, we can see that the improved schemes are significantly better than the basic schemes. The main reason is that the improved schemes can switch between different perturbation methods according to different moving characteristics, and always choose the most suitable perturbation method, while the basic schemes use only one perturbation method throughout the publication, and the perturbation cannot always suit the moving characteristics.

\begin{figure}
  \centering
  \subfigure[MAE]{
    \includegraphics[width=1.5in]{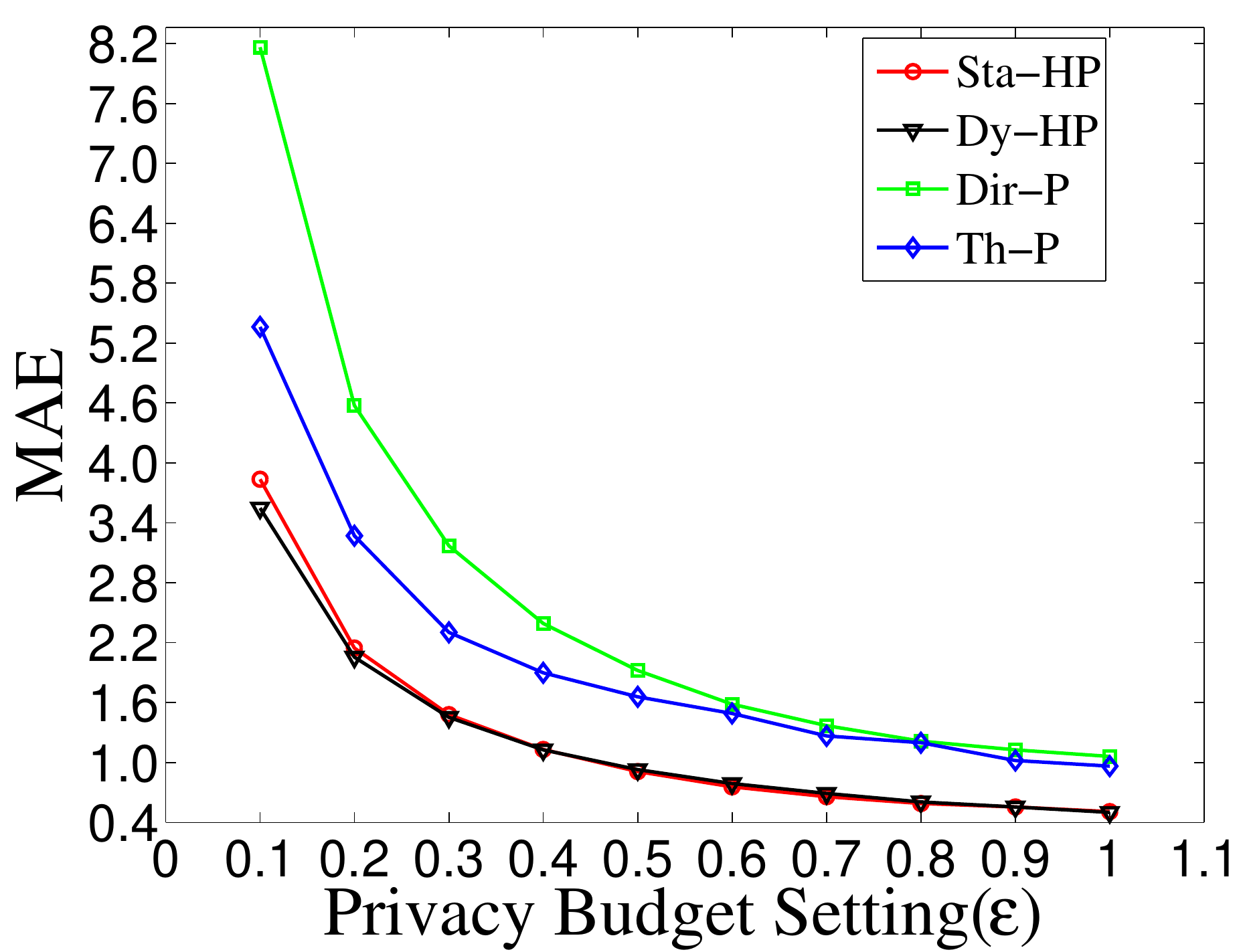}
  }
  \subfigure[MRE]{
    \includegraphics[width=1.5in]{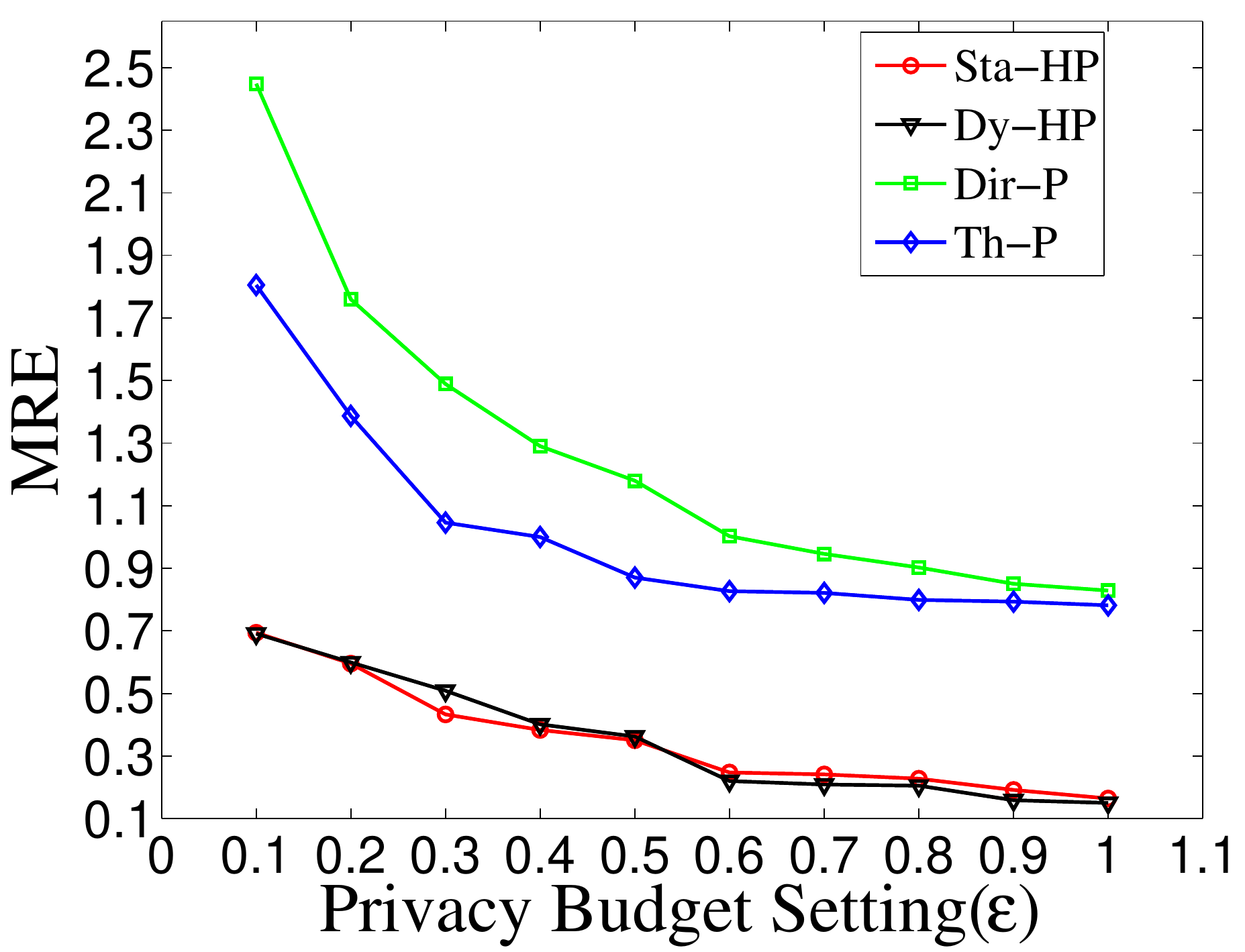}
  }
  \caption{Utility vs  $\epsilon$  on synthetic dataset}\label{fig:synthetic}
\end{figure}
\begin{figure}
  \centering
  \subfigure[MAE]{
    \includegraphics[width=1.5in]{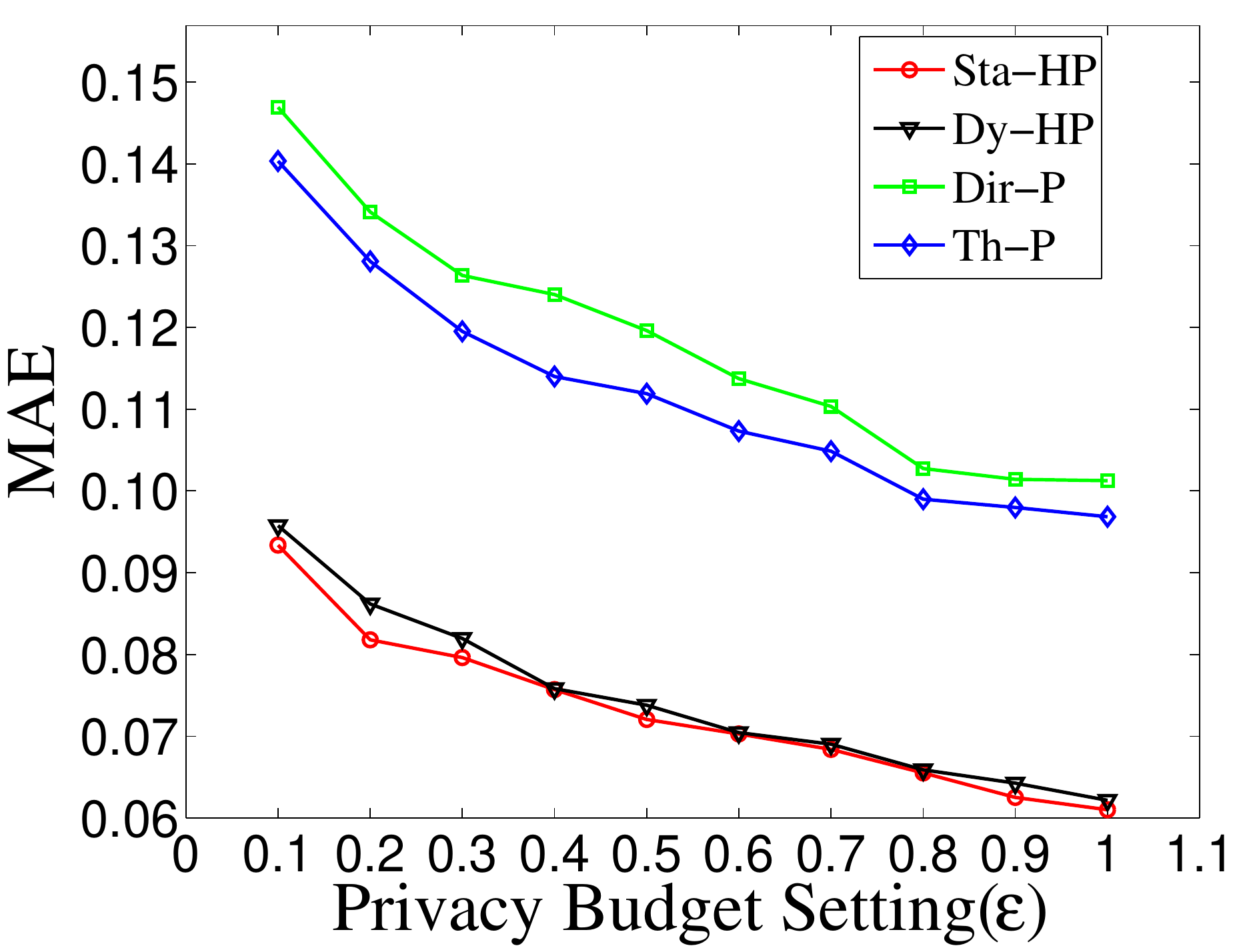}
  }
  \subfigure[MRE]{
    \includegraphics[width=1.5in]{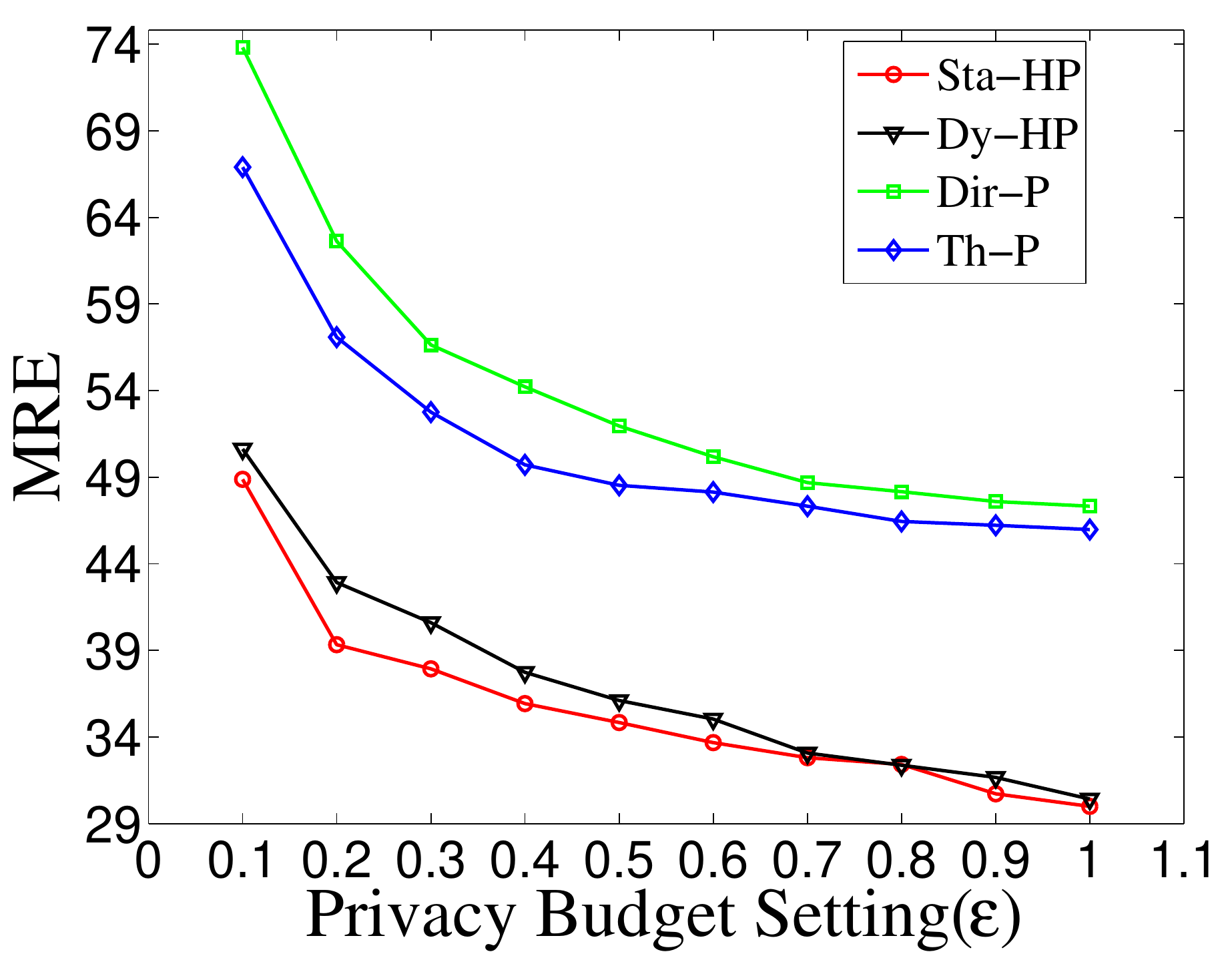}
  }
  \caption{Utility vs  $\epsilon$  on taxi dataset}\label{fig:taxi}
\end{figure}

\textbf{Effect of Post Process.} We conduct experiments of our schemes with and without post-process over two dataset to evaluate the effects of the consistency post-process mechanism. As shown in Fig.~\ref{fig:synthetic-post} and Fig.~\ref{fig:taxi-post}, we compare MAE and MRE values with and without post-process on the proposed schemes when the privacy budget $\epsilon$ changes from 0.1 to 1. On the synthetic dataset, we can find that  MAE and MRE values with and without post-process are gradually decreasing as the gradual increase of privacy budget $\epsilon$. However, the MAE and MRE values with post process are smaller than those without post-process. On the taxi dataset, this conclusion is similar. In addition, we can find that in the case when the privacy budget is small (eg., $\epsilon$ is 0.1), the post-process mechanism has a significant effect on reducing the error. And as the privacy budget becomes larger, the effect becomes less obvious. This is because when the privacy budget $\epsilon$ is large, the amount of noise introduced is small, and the error is also small.
\begin{figure}[t]
  \centering
  \subfigure[MAE]{
    \includegraphics[width=1.5in]{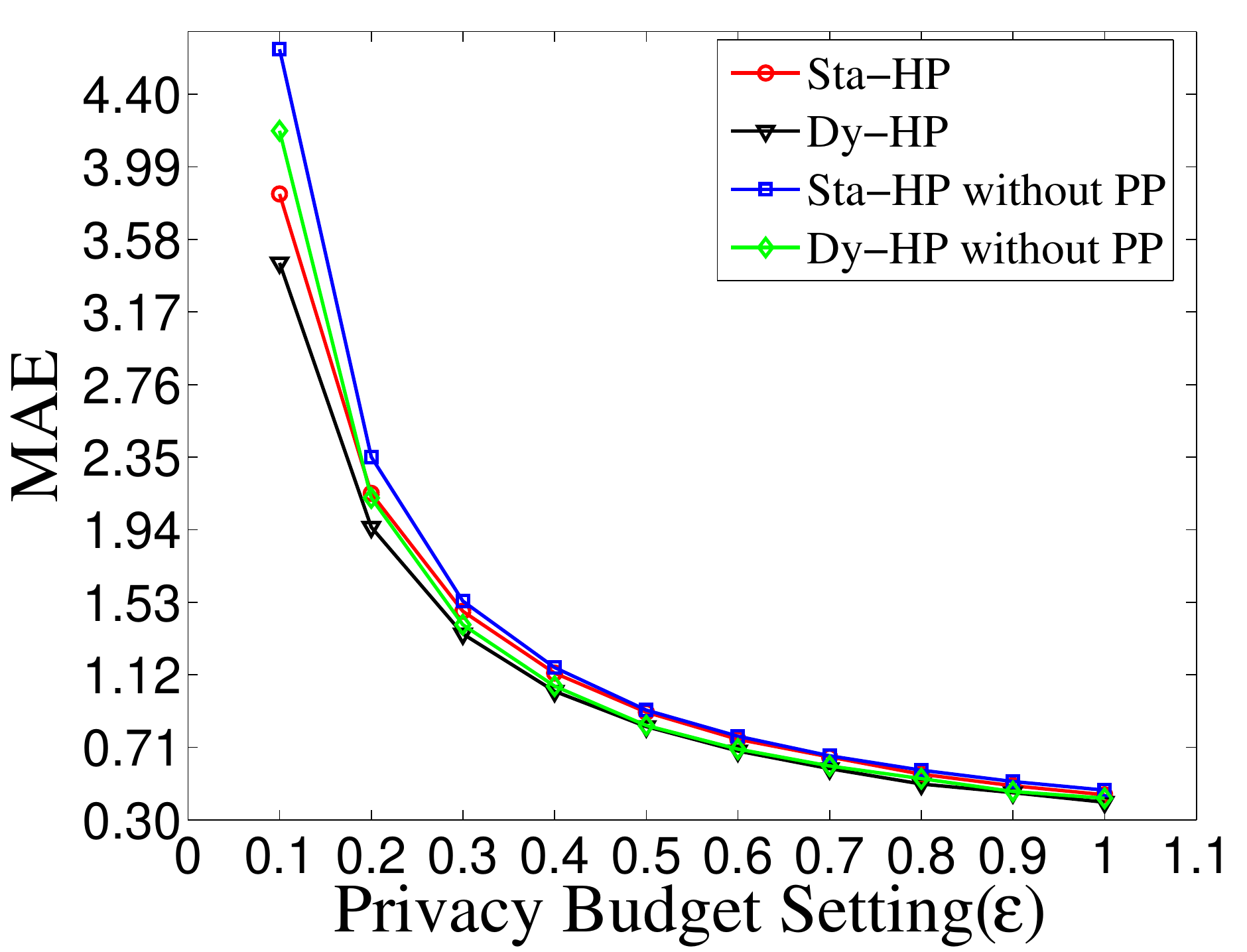}
  }
  \subfigure[MRE]{
    \includegraphics[width=1.5in]{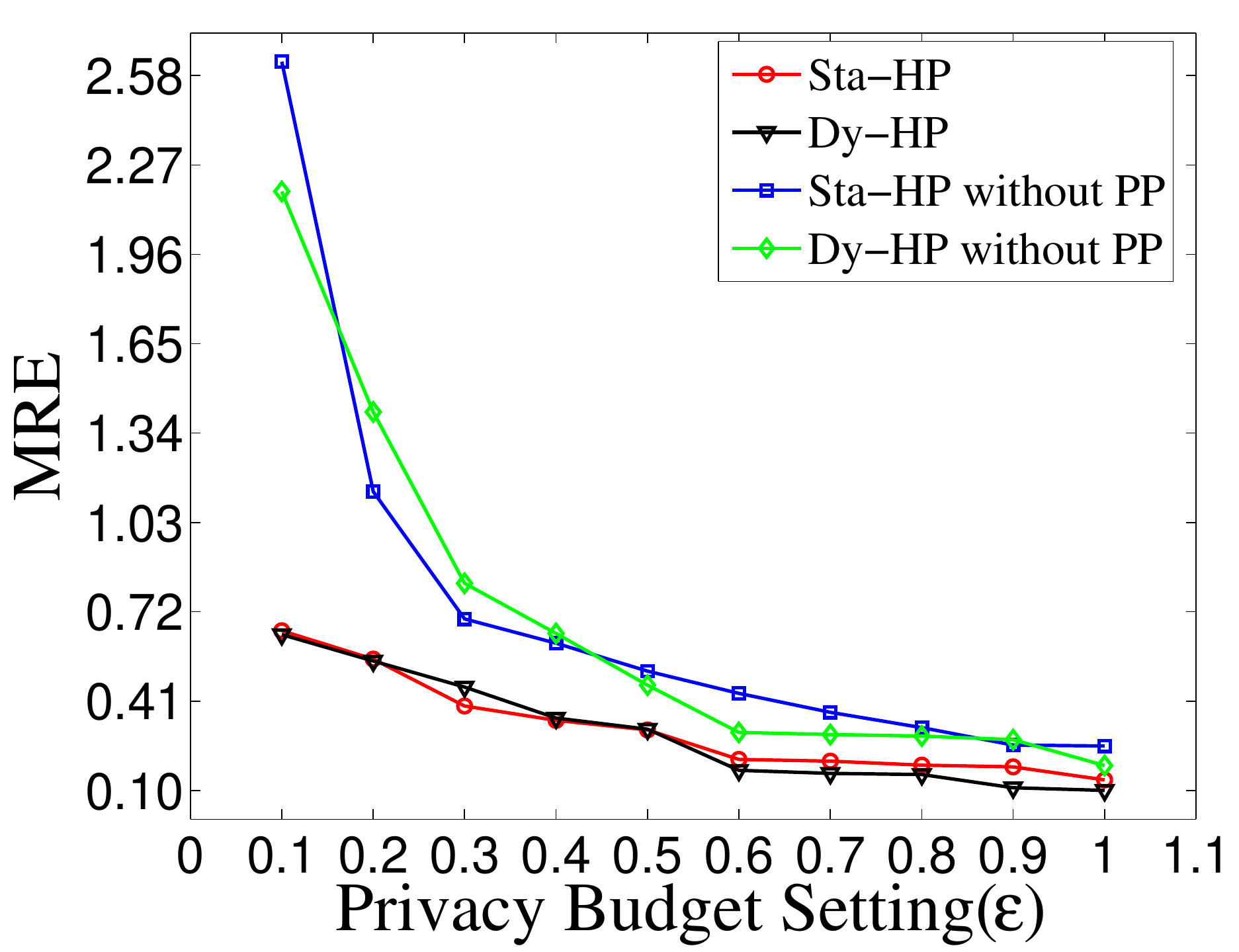}
  }
  \caption{Utility vs  $\epsilon$  on synthetic dataset}\label{fig:synthetic-post}
\end{figure}
\begin{figure}[t]
  \centering
  \subfigure[MAE]{
    \includegraphics[width=1.5in]{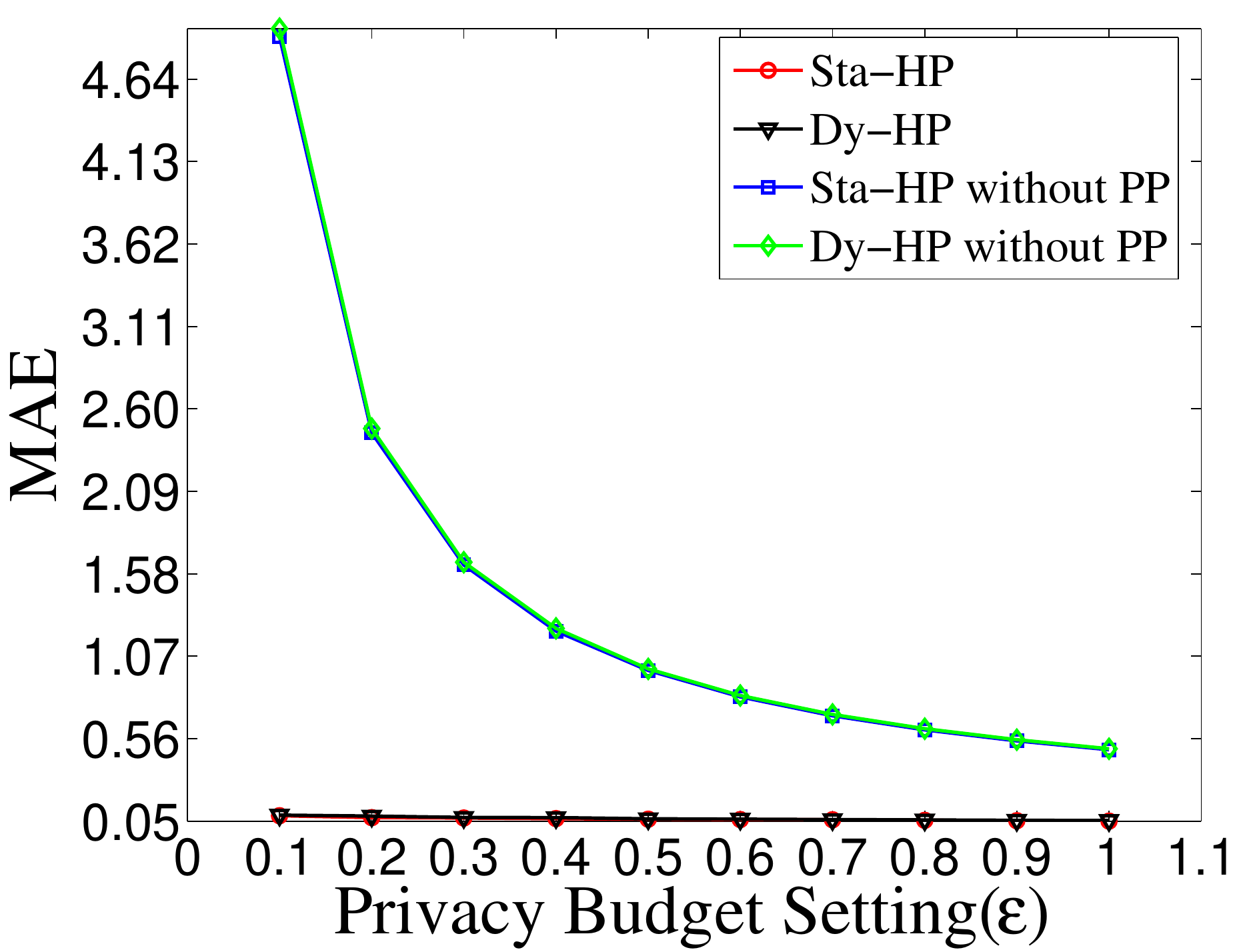}
  }
  \subfigure[MRE]{
    \includegraphics[width=1.5in]{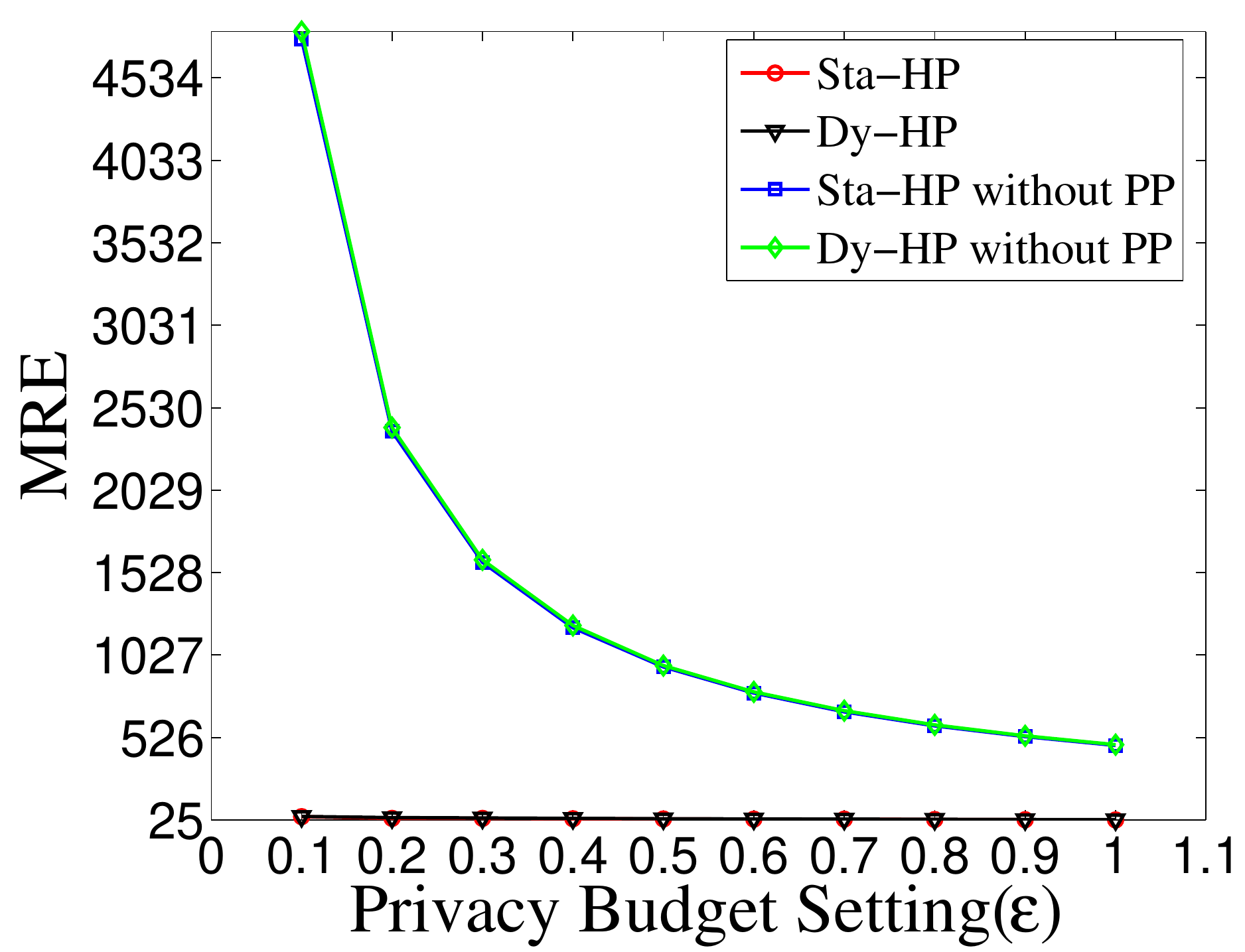}
  }
  \caption{Utility vs  $\epsilon$  on taxi dataset}\label{fig:taxi-post}
\end{figure}

\section{Conclusions}\label{sec:conclusions}
In this paper, we have propose several schemes for differentially private aggregated mobility data publication. We first propose two basic schemes and a post process mechanism to ensure the data validity and reduce the data utility, and discuss their drawbacks. Then, we examine the moving characteristics of mobile users, and propose two improved schemes for static and dynamic data publications, respectively, by combining the two basic schemes and leveraging the moving characteristics. Experimental results on real and synthetic datasets have shown that by taking advantage of the moving characteristics of aggregated mobility data, we can greatly improve the performance of the differentially private data publication. For the privacy reason, we have not obtained a dataset of real mobile users, and verified our schemes. We leave this as the future work.

\bibliographystyle{ACM-Reference-Format}
\bibliography{Refference}

\end{document}